\documentclass[a4,leqno,12pt]{article} % For LaTeX2e
\usepackage[a4paper,margin=3cm]{geometry}
\usepackage{amsmath, amsthm, amssymb, amsfonts}
\usepackage{enumerate}  

\usepackage{rotating}
\usepackage{hyperref}
\usepackage{url}
\usepackage{bbm}

\usepackage{graphicx}
\graphicspath{{./images/}}

\def\R{\mathbbm{R}}
\def\softmax{\text{\sf softmax}}
\def\RM{\text{\sf RM}}
\def\Id{\textbf{\sf Id}}

\def\Pr{\mathbb{P}}
\def\Esp{\mathbb{E}}
\def\CE{\mathbf{H}_{CE}}
\def\CS{\mathbf{C}_S}

\def\frakC{{\frak C}}

\def\calD{{\cal D}}

\def\calI{{\cal I}}
\def\calL{{\cal L}}
\def\calT{{\cal T}}
\def\calV{{\cal V}}

\let\toprule=\hline
\let\midrule=\hline
\let\bottomrule=\hline

\def\frakX{\mathfrak{X}}

\title{Do Word Embeddings %
Really Understand Loughran-McDonald's Polarities?}
\author{
Mengda Li
and 
Charles-Albert Lehalle 
}

\usepackage{graphicx}
\usepackage{subcaption}

\usepackage{amsthm}
\newtheorem{theorem}{Theorem}[section]

\newtheorem{property}[theorem]{Property}

\newtheorem{definition}[theorem]{Definition}

\begin{document}

% ===========================================================================
\maketitle

\begin{abstract}
    In this paper we perform a rigorous mathematical analysis of the word2vec model, especially when it is equipped with the Skip-gram learning scheme.
    Our goal is to explain how embeddings, that are now widely used in NLP (Natural Language Processing), are influenced by the distribution of terms in the documents of the considered corpus. 
    We use a mathematical formulation to understand how this family of language models captures the joined distribution of terms in sentences. This formulation shed light on how the decision to use such a model makes implicit assumptions on the structure of the language, in particular about the short range and long range dependencies between terms.
    We show how Markovian assumptions, that we discuss, lead to a very clear theoretical understanding of the formation of embeddings, and in particular the way it captures what we call \emph{frequentist synonyms}. These assumptions allow to produce generative models and to conduct an explicit analysis of the loss function commonly used by these NLP techniques, that asymptotically reaches a cross-entropy between the language model and the underlying true generative model.
    
    Moreover, having in mind to assess the identifiability of embeddings, we produce synthetic corpora with different levels of structure and show empirically how the word2vec algorithm succeed, or not, to learn them.
    It leads us to empirically assess the capability of such models to capture structures on a corpus of around 42 millions of financial News covering 12 years. That for, we rely on the Loughran-McDonald Sentiment Word Lists largely used on financial texts and we show that embeddings are exposed to mixing terms with opposite polarity, because of the way they can treat antonyms as frequentist synonyms.
    Beside we study the non-stationarity of such a financial corpus, that has surprisingly not be documented in the literature. We do it via time series of cosine similarity between groups of polarised words or company names, and show that embedding are indeed capturing a mix of English semantics and joined distribution of words that is difficult to disentangle.
    
\end{abstract}

%\tableofcontents 

\paragraph{Acknowledgments.}
Authors would like to thank Sylvain Champonnois for deep discussions about the nature of text polarity and biases of embeddings, Jean-Charles Nigretto for preliminary work on biases in doc2vec models, Elise Tellier for long discussions about the optimal size of embeddings, Laurent El Ghaoui for challenging discussion about the representation of language models by embeddings, and G\'erard Ben Arous for discussions on the convergence of criteria and the choice of \emph{observables} in the empirical analysis of word embeddings.

% ===========================================================================
\section{Positioning of this work}

Recent advances on Natural Language Processing (NLP) are largely based on the use of embeddings. It started with performing SVD on vectorial representation of words (see for instance \cite{laham1998learning}) before turning to less linear analysis like the use of Self Organizing Maps \cite{lagus1999websom}. More recently, the term of ``embeddings'' emerged with the idea of learning a lookup table while performing downstream tasks \cite{2013arXiv1310.4546M}; later on downstream tasks used LSTM (Long Short Term Memory) networks to account for the sequential aspects of the natural languages \cite{ghosh2016contextual}. 
Now BERT-like models as other attention-based models proposed to plug embeddings in sequences of local LSTM structures, to obtain more contextual (i.e. bi-directional) Euclidean representations of terms \cite{vaswani2017attention,devlin-etal-2019-bert}.

Despite the apparent successes of this series of improvements, very few mathematical modeling has been proposed. 
Some interesting papers are descriptive, exposing the formulas behind the mechanisms of embeddings, like \cite{implicitmatrix} that investigate a SVD viewpoint on embeddings, but they do not provide a probabilistic model corresponding to the implicit assumptions made by the considered NLP modeling algorithm.
As a comparison, the literature on topic extraction naturally provides this kind of generative models \cite{blei2003latent}.

The goal of this paper is to provide such a mathematical understanding of how and what embeddings are learning. Since it is a complicated task, we are limiting ourselves to the word2vec model, that is paradigmatic of this family of language models. The sophistication of subsequent approaches (like LSTM or BERT), is ``simply'' to add some sequential memory and some ``locality'' (since they use an ensemble of embeddings, each of them specializing locally, i.e. in the neighborhood of a syntactic context). Writing the formulas for such models is probably possible, but would require a lot of sophisticated notations. We clearly restrict our analysis to global embeddings of the word2vec family, but cover both Skip-gram and CBOW modeling approaches.
\medskip

During this theoretical analysis, we define what we call a \emph{Reference Model}. Qualitatively speaking, it is a very large stochastic matrix collecting for each word of the vocabulary the probability of occurrence of all words in ``its neighborhood''. Given this Reference Model, no embedding is needed (or one can consider that the associated embedding is the identity matrix). It allows us to define a word2vec representation of a corpus as a compression of its Reference Model in two components: one the one hand an embedding, qualitatively mapping ``similar words'' (in the sense of the Reference Model) together, and hence defining ``frequentist synonyms''; and on the other hand a context matrix, allowing to recover the Reference Model once it is multiplied by the embedding (and after a slightly non linear operation, i.e. a softmax layer). 

% not cited because arxiv only: https://arxiv.org/pdf/2008.07720.pdf + https://arxiv.org/pdf/2003.07278.pdf
We are here in the spirit of other papers presenting embeddings as a compression (see for example \cite{may2019downstream,raunak2019effective,acharya2019online}), along with the initial papers on SVD, except that we can prove, under some assumptions, that the minimized criterion is a cross-entropy between the Reference Model and the compressed representation made of the embedding and the context matrices.
\medskip

In a second stage, we explore further the concept of frequentist synonyms (that ``should be put in the same neighborhood by the compression'') in association with the identifiability of word2vec embeddings. That for we generate synthetic corpora with know properties thanks to the control of their Reference Model. It allows us to test empirically that if the word2vec model can recover part of the structure we inject in the synthetic texts. The results are bellow standard expectations: the embeddings of synonyms do not exhibit strong cosine similarities, despite the fact that we use generative models satisfying strong mathematical assumptions (that is not the case of natural language).
\medskip

Last but not least, we use this understanding to train embeddings on a corpus of 42 millions of financial news from 2008 to 2020. 
The literature on the use of NLP on financial data has been pushed forward by the work of Loughran and McDonald who analysed the polarity of 10-K fillings of listed company \cite{loughran2011liability}. One of the outcome of their work is a dictionary of polarized words in categories (Positive, Negative, Litigious, etc) that have been later used as a supervised dataset \cite{kumar2017iitpb}, as a benchmark \cite{ke2019predicting} or as an input \cite{li2014news} to analyze other corpora (financial news being one of the most used). It is clear that NLP is now part of the toolbox to predict stock returns \cite{xing2018natural,araci2019finbert}, but to authors' knowledge no systematic study analyzing the capability of language models to capture financial semantics have been proposed. In this paper we do not try to predict stock returns, we rather focus on providing insight about the structure of financial news from the perspective of embeddings, with a focus on semantic antonyms and non-stationarity. 
\medskip
% too generic: \cite{gentzkow2019text}

The paper is structured as follow: Section \ref{sec:theory} starts by defining properly word2vec model and its Skip-gram version, then explains the relationship between the internal representation of such an language model and the underlying ``truth'' of a Reference Model. Since a Reference Model is target that embeddings try to learn, it is pivotal in our analysis. We then make additional Markovian assumptions to ease the process of proving how the standard Skip-gram loss function asymptotically behaves like a cross-entropy between the learned model and the Reference Model.
This allow us to comment the influence of synonyms on the structure of embeddings.
Section \ref{sec:synth:experiments} exhibits experiments on synthetic corpus generated under restrictive assumptions, allowing us to control what the embeddings should learn. It shows the importance of the structure of the underlying model, and then it shows that the similarity between the representation of synonyms is not as large as expected, but significant when compared to other groups of words.
Then Section \ref{sec:finance} exploits a corpus made of 42 millions of financial News in conjunction with the Loughran-McDonald Lexicon. The latter is providing the structure of polarized words (mainly: Positive, Negative and Litigious) needed for our analysis. It allows us to explore the way embeddings can mix synonyms and antonyms, and we find and explain different results obtained on headlines only versus on the full-text of financial News. 
Moreover, we study the non-stationarity of this structure, in particular around one particular example of a company name that is listen under Wikipedia's \emph{List of corporate collapses and scandals}.

\section{Theoretical analysis: learning embeddings and its relation with generative models}
\label{sec:theory}

An embedding is an vectorial representation of terms of a document, that is learned jointly with other NLP tasks that is learned to model a language or more recently jointly with different downstream task (like part of speech tagging or entity recognition). The underlying idea is that the learning phase positions the terms in the space of embeddings such a way that the downstream tasks are easier.
As a consequence, it is expected that the embeddings will position in the same neighborhood terms having similar roles in sentences. Academic papers, working of the topology in the space of embedding, identified some interesting algebra properties in the space of embeddings, like: ``\emph{King - man + woman = Queen}''. In this paper we focus on the word2vec algorithm which is proposed in \cite{2013arXiv1301.3781M} and \cite{2013arXiv1310.4546M}, and especially on its Skip-gram version.  

\subsection{The word2vec model}

To be compatible with the standard word2vec Skip-gram notations of \cite{word2vec},
let assume that we face a set of documents, or a very long document that is a sequence of words belonging to a vocabulary $\calV:=\{x_1, \cdots, x_V\}$.
The word2vec proposed in \cite{Goodfellow-et-al-2016} can be written as a neural network with one hidden layer of size $N$, that is called the \emph{embedding size} and that is usually far lower than $V$ (see Figure \ref{fig:Skip-gram}).
Typically English vocabulary needs at least 70,000 words and the embedding size used by practitioners is between 150 and 400.

To feed to the neural network, the $i$-th word of the vocabulary is encoded using a one-hot vector $x_i$, i.e. $x_i\in\R^V$ that is a vector with zeros everywhere but only one $1$ at its $i$-th coordinate.
The weights between the input layer and the output layer can hence be represented by a $V\times N$ matrix ${W}$. Each row of ${W}$ is the $N$-dimension vector representation ${e}_i$ of the associated word of the input layer. Formally row $i$ of ${W}$ can be identified with the representation of $i$-th word and noted ${x_i}^T W$, given as the context word one-hot vector ${x}_i$. 
The activation of the hidden layer that is linear: ${h}_{x_i}^T = {x_i}^TW = {W}_{(i, \cdot)}$,
which is essentially copying the $i$-th row of ${W}$ to ${h}_i$. In this paper, we use row vector notation. We call ${W}$ the \emph{word-embedding matrix}, and its $i$-th row is commonly called the \emph{embedding of the $i$-th word of the vocabulary}.

From the hidden layer to the output layer, there is a matrix of weight commonly noted ${W}'$, which is an $N\times V$ matrix, associating an embedding to a long row of size $V$. $W'x_j$ the $j$th column of the \emph{context matrix} ${W}'$. Its output layer as a $\softmax$ activation, mapping a vector $(z_i)_i$ to another vector $(\exp(z_i)/\sum_j \exp(z_j))_i$ that is positive and which coordinates sum to one. The output of such a neural network is thus homogeneous to a vector of probabilities. As a consequence, the word2vec allows to compute the expected probability of occurrence of any $j$-th word of the vocabulary conditioned on the input word $i$ that reads $\softmax({x_i^T W W'}) x_j$.

\begin{figure}[htbp]
\centering
\includegraphics[width=3in]{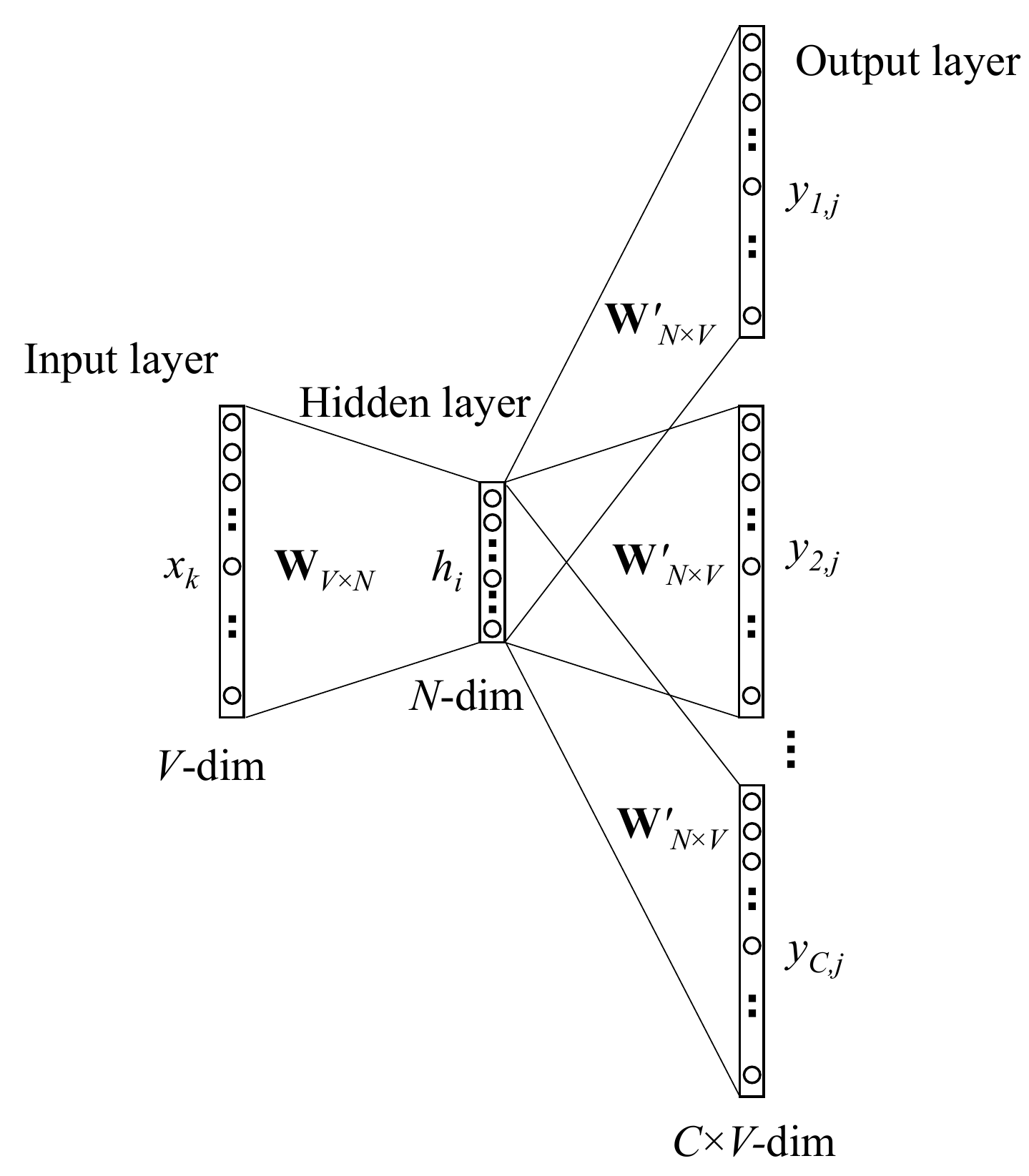}
\caption{The Skip-gram model.}
\label{fig:Skip-gram}
\end{figure}

Moreover, the word2vec Skip-gram model has an hyperparameter $C$ that is the number of words that is explicitly considered by the model. 
In this paper, for the simplicity of notation, the considered words are the $C$ words following the input word.\footnote{It is not need in general; their positions can be chosen by the user.} 
As a consequence the variable
\begin{equation}
y_{c,j}^i =  \softmax({x_i^T W W'}) x_j, \quad \quad \forall c \in \{1,.., C\}
\label{eq:cbow-uj}
\end{equation}
computes the estimate of the probability of occurrence of the $j$-th word of the vocabulary amongst the $C$ words following the $i$-th word of the vocabulary in the considered corpus of documents.
Note that each output $y_{c,j}^i $ is independent and only depend on the input word.

\begin{definition}[Parameters of a word2vec model]
A word2vec model is defined by the triplet of parameters: embeddings, contexts and width; i.e. $(W,W',C)$.
\end{definition}

\subsection{The Skip-gram loss function}  

\paragraph{Assuming the existence of an ``underlying'' model of the text.}
In the following sections, we use $(X_k)_k$ to denote the sequence of words in a corpus. 
To be able to define a probabilistic description of the embeddings, one needs first to make the weak assumption that the sequence of the $K$ words constituting the corpus are drawn from a probability distribution $\calL_\calT$ that exists.
You can think that this distribution allows to generate the whole sample at once, with a lot of bi-directional dependencies, or you can think that $(X_k)_k$ is a stochastic process (meaning that there is an ``information arrow'' pointing from the left to the right in the text). It may be generated by a Markov chain if you believe in very short term memory in the text, or even be i.i.d. realizations of a random variable if you think the words are generated according to their histogram of frequencies in the English language. 
The existence of this law will enable us to write probabilities.
\\
Since equality (\ref{eq:cbow-uj}) defines the likelihood of the occurrence of word $j$ in the neighborhood of word $i$, the word2vec model implies a specific structure of this probability $\calL_\calT$. Or at least it assume that the loss function it minimizes pushes this likelihood to be compatible with the underlying distribution.
This is typically what we will explore in this paper: which assumptions on $\calL_\calT$ are naturally compatible with the way the word2vec model build its likelihoods? That for we need to start with a good understanding of the Skip-gram loss function.
\medskip

Following our notations: $X_k$ is the $k$-th word in the document while $x_i$ is the $i$-th word in the vocabulary. $x_i$ denotes the one hot encoding of the $i$-th word in the vocabulary (word index). Hence the event ``$X_k=x_j$'' means that the $k$-th word observed in the corpus is the $j$-th word of the vocabulary.

\paragraph{Writing the Skip-gram loss function.}
The part of the loss function associated with the $k$th word of the document corresponds to the estimate by the word2vec model of likelihood to observe the $C$ following words of the corpus. It reads
\begin{eqnarray}\label{eq:loss:skipgram}
\ell(X_k,\ldots,X_{k+C}) 
&=&{-\log \hat\Pr_{W,W'}(X_{k+1}=x_{j_1}, \cdots, X_{k+C}=x_{j_C} | X_k=x_i) }\\\nonumber
&=& -\log \prod_{c=1}^C\softmax(x_i^T W W')x_{j_c} \\\nonumber
&=& -\log \prod_{c=1}^C\frac{\exp(x_i^T W W'x_{j_c})}{\sum_{j'=1}^V\exp(x_i^T W W' x_{j'})} \\\nonumber
&=& -\sum_{c=1}^C x_i^T W W'x_{j_c} + C\cdot\log\sum_{j'=1}^V\exp(x_i^T W W' x_{j'}),
\end{eqnarray}
where $x_{j_1}, \cdots, x_{j_C}$ are the word occurring after the $k$-th one in the text. During the learning phase, the Skip-gram word2vec hence deals with sequences of $C+1$ words that can be found in the text.
\\
The full loss function is the average of $\ell(\cdot)$ over all the words of the text, that can be expressed as $\frac{1}{K}\sum_k \ell(X_k,\ldots,X_{k+C})$, i.e. the empirical average over the words of the text.

\paragraph{Defining a Reference Model.}
%  and removing the softmax layer
First define what we will call a ``\emph{Reference Model}'': qualitatively it is a word2vec model with a trivial embedding, i.e. the size of the embedding is the size of the vocabulary: $N=V$, and hence the embedding matrix $W$ is the identity.
Since the rows of the context matrix of a Reference Model do not interact and the $\softmax$ layer is not needed; it is enough to arbitrary scale its rows to sum to one. This non-linearity can thus be removed from the word2vec of a Reference Model.

\begin{definition}[Reference Model]\label{def:RM}
A \emph{Reference Model} $\RM$ is a triplet $(\Id, W'_0,C)$ that can be identified to the parameters of a word2vec neural network. 
With $V$ the size of the vocabulary, $W'_0$ is a $V\times V$ stochastic matrix which element $(i,j)$ records the probability to see the $j$-th word of the vocabulary in the $C$ words following the $i$-th word.
\end{definition}
Different kinds of Reference Models can be built:
\begin{itemize}
    \item[($i$)] The Reference Model corresponding to a given word2vec model $(W, W', C)$ that is $(\Id, W'_0, C)$ where for any $i$, $x_i^T W'_0:=\softmax(x_i^T W W')$; we use the notation $\RM(W,W')$.
    \item[($ii$)] The Reference Model corresponding to a corpus of document that is $(Id, W'_0, C)$ where $W'_0$ is this time made of the averaged empirical probabilities of occurrence of the $C$ words following each possible word;
    we use the notation $\RM(X_1, \ldots, X_K)$.
    \item[($iii$)] The Reference Model corresponding to underlying probability distribution $\calL_\calT$ that generated the corpus;
    we use the notation $\RM(\calL_\calT)$.
\end{itemize}
Note that for the same corpus: 
\begin{itemize}
    \item Given $(X_1, \ldots, X_K)$ has been generated by $\calL_\calT$ (and provided that $K$ is large enough), we expect to have $\RM(X_1, \ldots, X_K)\simeq \RM(\calL_\calT)$.
    \item If a word2vec is a ``good model'' for this corpus, the corresponding Reference Model $\RM(W,W')$ should be ``close'' to the third one $\RM(\calL_\calT)$.
\end{itemize}
The question of a \emph{generative model} can be formulated this way: 
\emph{what are the assumptions over $\calL_\calT$ \emph{(i.e. the true underlying distribution)}, such that the Reference Model of the word2vec learned on a \emph{(very long)} sample generated by $\calL_\calT$ is the same as the Reference Model of $\calL_\calT$?}

\subsection{Probabilistic weaknesses of the word2vec embeddings}
\label{sec:PLM}

% https://towardsdatascience.com/deconstructing-bert-part-2-visualizing-the-inner-workings-of-attention-60a16d86b5c1

It is now clear that the information contained in a word2vec model on the learned language is captured by a factorization over $V\times N + N\times V$ coefficients (i.e. the embeddings $W$ and the context $W'$) of the joined probability of occurrences of $C+1$ words at a given distance that could be described by $(V\times V)^C$ parameters.
The key quantity that is manipulated by the word2vec is an estimate of $\Pr_{\calL_\calT}(X_{k+1}=x_{j_1}, \cdots, X_{k+C}=x_{j_C} | X_k=x_i)$ for any word position $k$ in the observed text. We will use the notation $\hat\Pr_{W,W'}(X_{k+1}=x_{j_1}, \cdots, X_{k+C}=x_{j_C} | X_k=x_i)$ to underline that this probability is estimated using matrices $W$ and $W'$ only.

We will list here the main weaknesses of this model and give pointers and intuitions to how more recent models answered to these points. 

\paragraph{Weakness 1: The input is made of only one word.}
The first potential issue of the word2vec Skip-gram approach is that there is only one word as input. It implies that once the model parameters $W$ and $W'$ are chosen, there is no difference between $\hat\Pr_{W,W'}(X_{k+1}=x_{j_1}, \cdots, X_{k+C}=x_{j_C} | X_k=x_i, X_{k-1}=x_{i'})$ and $\hat\Pr_{W,W'}(X_{k+1}=x_{j_1}, \cdots, X_{k+C}=x_{j_C} | X_k=x_i, X_{k-1}=x_{i''})$.\footnote{If we consider the learning process it is clear that the loss function had to cope with the following two sequences: $(X_{k-1}=x_{i'}, X_k=x_i, X_{k+1}=x_{j_1}, \cdots, X_{k+C-1}=x_{j_{C-1}})$ and $(X_{k-1}=x_{i''}, X_k=x_i, X_{k+1}=x_{j_1}, \cdots, X_{k+C-1}=x_{j_{C-1}})$. The information is inside $W$ and $W'$, but at this stage it is difficult to know how, and it is sure that it is averaged using as weights the relative number of occurrences of these sequences in the text.}
They are a lot of different ways to explicitly fix this issue about a potential fragility of the input: 
\begin{itemize}
    \item the \emph{CBOW} version of the word2vec algorithm is one of them.
    The CBOW uses the same weights $W$ and $W'$ but instead of minimizing a loss that is the likelihood of occurrence of a list of words observed ``after'' $X_k$, it takes a list of word occurring``before'' $X_k$ and averages their representations to target the $k+1$ word of the text:
    $$\begin{array}{l}
    \hat\Pr^{\text{\sc cbow}}_{W,W'}(X_{k+C}=x_{j_C} | X_{k+C-1}=x_{j_{C-1}}, \cdots, X_{k+1}=x_{j_1}, X_k=x_i) \\[.5em]
    \phantom{\hat\Pr^{\text{\sc cbow}}_{W,W'}}= \softmax( \frac{1}{C}(x_{j_{C-1}}+ \ldots + x_{k+1}+x_k)^TW W' )x_{j_C}.
    \end{array}$$
    It means that the internal representations of the $C-1$ input vectors are averaged, and that this average is used to predict the likelihood of the next word.
    
    \item Since Skip-gram is trying to estimate the averaged probability of occurrence of words occurring after the pivotal $k$-th word, and the CBOW is doing the reverse (estimate the probability of occurrence of the $k+1$-th word given an average of the word that are before $X_k$), one could imagine a mix of these two approaches:
    using a weighted average of word ``before'' $X_k$ to predict a weighted averaged probability of words positioned after $X_k$.
    
    It is not very far away of what the self-attention mechanism is doing \cite{vaswani2017attention} except that weights for the average are not chosen a priori, but they are learned and they are a function of words surrounding each weighted word: it is the \emph{attention} associated to the occurrence of the sequence of words $(X_{k-1}=x_{i'}, X_k=x_i, X_{k+1}=x_{j_1}, \cdots, X_{k+C-1}=x_{j_{C-1}})$.
    
    \item Another way to have more information about the words before $X_k$ would be to increase the size of the state space and to concatenate $X_k$ and $X_{k-1}$. But it is considered to be far too demanding in number of parameters. NLP methods prefer to use averages.
\end{itemize}
The take-away if this weakness is that it prevents a word2vec to capture some subtleties of the learned corpus. 
One solution is to average over more words such a way that their weights in the average (i.e. attention) are conditioned by the joined distribution of words. It certainly produces \emph{regularization} and \emph{localization} (in the sense that the embeddings are weighted differently according to their position in the very high dimensional space of observed sequences of words). It means that the conclusions we will obtain on word2vec will probably be only valid locally for attention mechanisms.

\paragraph{Weakness 2: The ordering of words is not taken into account.}
From the previous Section, it is clear that for a Skip-gram word2vec (and we know it is the same for a CBOW one) once the parameters $W$ and $W'$ are chosen, there is no difference between the word2vec estimate of $\hat\Pr_{W,W'}(X_{k+1}=x_{j_1}, \cdots, X_{k+C}=x_{j_C} | X_k=x_i)$ and $\hat\Pr_{W,W'}(X_{k+1}=x_{\sigma(j_1)}, \cdots, X_{k+C}=x_{\sigma(j_C)} | X_k=x_i)$, for any $C$-permutation $\sigma$.

Nevertheless, if we consider the learning process: only sequences that have been seen in the text are considered in the minimization of the loss function, hence $W$ and $W'$ incorporate a trace of the sequences. 
But it would be interesting to explicitly inject the positions of the words in the probabilistic model.

This is what BERT is doing, by concatenating the embeddings of the words (i.e. the $W$ matrix of the word2vec) with a positional embedding \cite{devlin-etal-2019-bert}.
It is clear that it adds the positional information that is missing in word2vec. 
As an illustration, just have a look at the expected effect on the estimate of the likelihood if this idea of ``positional embeddings'' is transposed to the word2vec mechanism:
say that we augment the two matrices $W$ and $W'$ by a positional one $P$ that is used the standard word2vec way. Now we can replace (with $C=2$, to keep it simple, and with the notation $p_u$ to encode that a word is in position $u$)
$$\hat\Pr_{W,W'}(X_{k+1}=x_j, X_{k+2}=x_{j'}|X_k=x_i) = \frac{1}{2}(\softmax(x_{i}^T W W')x_j+\softmax(x_{i}^T W W')x_{j'})$$
by:
$$\begin{array}{l}
\hat\Pr_{W,W',P}(X_{k+1}=x_j, X_{k+2}=x_{j'}|X_k=x_i) \\[.3em]
\phantom{\hat\Pr_{W,W',P}} = \displaystyle \frac{1}{2}\left(\softmax(x_{i}^T W W' + p_1^T PP'p_2)x_j+\softmax(x_{i}^T W W' + p_1^T PP'p_3)x_{j'}\right).
\end{array}$$
This produces different shifts inside the softmax depending on the relative positions of the words. It implies that using such a mechanism, a word2vec would consider two words to be interchangeable with respect to a given $x_i$, not only if they have the same $x_i^TWW'$ but also if they are positioned at the same distance if $i$. 
% One could keep this in mind reading the section of this paper about the structure of language models, and especially synonyms and anti-synonyms.

\subsection{Markovian generative models}

In this section, we explore the theoretical properties of the word2vec embeddings under the assumption that the underlying language text is generated by a Markov model.
We restrict ourselves to the Skip-gram approach: the loss function that is minimized is derived form $\ell(k)$ defined by equality (\ref{eq:loss:skipgram}).

\begin{definition}[Markov generative model for texts]
A Markov generative model for a text over a vocabulary $\calV$ of size $V$ is defined by a stochastic transition matrix $K$ of size $V\times V$ such that $K_{i,j}$ is the probability that the $j$-th word of the vocabulary follows the $i$-th word of the vocabulary.

Thanks to $K$, and provided that an initial distribution $m_0$ over words is defined, a text of size $T$ is a stochastic process $(X_k)_{1\leq k\leq T}$ such that $X_0\sim m_0$ and 
\begin{equation}\label{eq:probaK}
\Pr(X_{k+1}=x_j|X_k=x_i)=K_{i,j}.    
\end{equation}
Markov generative models are restricted to irreducible Markov chains.
\end{definition}

It is natural restrict Markov generative models to irreducible Markov chain, since it is realistic to think that a language can link two arbitrary words of the vocabulary by a finite text.

Texts generated by Markov models are of course poorer that standard English texts. Nevertheless we could easily imagine ``multiple inputs'' or ``local'' versions of Markov text generation that would correspond the the formerly identified probabilistic weaknesses of the Skip-gram word2vec.
Moreover, one could think about extensions of the Markovian framework to more subtle ones like the one chosen for the Latent Dirichlet Allocation \cite{LDA} without killing most of the theoretical mechanisms that will be used in this section; but this is out of the scope of this paper.

Since a Markovian text satisfies $\Pr(X_{(k+c)}=x_{j}|X_k =x_{i}) = x_{i}^T K^c x_{j}$, it is straightforward to see that the main quantity of interest for a Skip-gram word2vec is linked to the Reference Model $\RM(K)$ corresponding to the Markov kernel $K$ since
\begin{equation}\label{eq:MK}
    % W_0\, W'_0 = 
    W'_0(\RM(K)) = \frac{1}{C}\sum_{i=1}^C K^i.
\end{equation}
This comes from the fact that the language model is trained on a sequence of $C$ words without any consideration of ordering.
Hence, given that $\calI$ is a uniform random variable over $\{1,\ldots,C\}$, the Skip-gram word2vec tries to estimate
\begin{eqnarray}\nonumber
    \Pr(X_{k+\calI} | X_{k}) %&= \sum_{i=1}^C  P(X_{k+\mathcal{I}} ,\mathcal{I}= i |  X_{k}) 
    %= \sum_{i=1}^C  P(X_{k+i}, \mathcal{I} = i |  X_{k}) \\
    &=& \sum_{i=1}^C   \Pr(X_{k+i} | X_{k}) \, \Pr( \calI= i) \\%\text{      by independence} \\
    &=& \frac{1}{C} \sum_{i=1}^C \Pr(X_{k+i} | X_{k}).
\end{eqnarray}

\subsubsection{Markov chain properties}

In this subsection, we formulate well known properties of Markov chains \cite{revuz2008markov} that we will need thereafter within our models and notations.
Let $\calD(T)=(X_1,X_2,\ldots,X_T)$ be a document made of $T$ words generated using a Markov chain with a transition matrix $K$ on a vocabulary of $V$ words, and denote by $\mu$ the stationary distribution of the Markov chain. The state space of our Markov chain is the vocabulary $\mathcal{V}$ which is finite. Hence our Markov chain is positive recurrent; we will largely rely on the ergodic theorem.

\begin{property}[Convergence of empirical distribution]
    \begin{equation}\label{eq:station:conv}
        \lim_{T\rightarrow \infty} \frac{1}{T}\sum_{k=1}^T \mathbbm{1}(X_k=x_1) = \mu(x_i) 
    \end{equation}
\end{property}
Note that $\mu$ is a vector of size $V$, and the $i$th coordinate $\mu_i = \mu(x_i)$ is the probability of appearance of work $i$ in the stationary distribution.
\begin{proof}
Direct consequence of Ergodic theorem. 
See Theorem 4.16 in \cite{LevinPeresWilmer2006} or Corollaire 13.6.2 in \cite{legall}.
\end{proof}

Then we state a property for the convergence of empirical conditional distribution:
\begin{property}[Convergence of empirical Conditional distribution]
Let $\mathcal{I}$ be an independent uniformly random index from $1$ to $C$ (each index has equal probability $\frac{1}{C}$),
\[
\lim_{T\rightarrow \infty} \frac{\sum_{k=1}^{T} \mathbbm{1}(X_{k}=x_i, X_{k+\mathcal{I}} = x_j)}
   {\sum_{k=1}^{T} \mathbbm{1}(X_{k}=x_i)} = \frac{1}{C}\sum_{l=1}^C K^l(x_i, x_j)
\]
\end{property}
\begin{proof}
We only stage the proof in the simplest $C=1$ case to keep it short; when $C>1$, the proof is similar.
Just observe that $(X_k, X_{k+1})_{k\in \mathbb{N}}$ is an irreducible Markov chain too over the set of states $\{(x,y) | K(x,y)>0\}$ with kernel $Q((a,b),(c,d))=K(c,d)\mathbbm{1}(b=c)$. Its unique invariant probability measure is $\pi(x,y) = \mu(x) K(x,y)$. Applying the ergodic theorem to $\mathbbm{1}(X_{k}=x_i, X_{k+1} = x_j)$ reads
\[
\frac{\frac{1}{T}\sum_{k=1}^{T} \mathbbm{1}(X_{k}=x_i, X_{k+1} = x_j)}
   {\frac{1}{T}\sum_{k=1}^{T} \mathbbm{1}(X_{k}=x_i)} \underset{T\to \infty}{\longrightarrow }  
   \frac{\mu(x_i) K(x_i,x_j)}{\mu(x_i)} = K(x_i,x_j)
\]
\end{proof}

\subsubsection{From a Reference Model to a Markov kernel}

Given a Reference Model with $C=1$, i.e. $(\Id, W_0',1)$, 
a Markov chain with $K:=W'_0$ is the corresponding Markov generative model:
by construction the two reference Models coincide. 
But when $C>1$, there does not always exist a $K$ satisfying (\ref{eq:MK}) for a given $W'_0$.
\begin{theorem}[Representative generative model]
Given a reference model $(\Id, W_0',C)$ on a vocabulary of $V$ words, 
there exists a Markov chain with a transition matrix $K$ verifying equality (\ref{eq:MK}) if one of these two conditions is verified:
\begin{enumerate}[(i)]
    \item $C=1$;
    \item $W'_0$ is symmetric (or diagonalizable) and all its eigenvalues are in $[0, 1]$.
\end{enumerate}
\end{theorem}
If our kernel $K$ is diagonalizable (for example, if it is reversible), then $K = P \Delta P^{-1}$ (and if it reversible, $P$ is a orthonormal basis, i.e., $P^{-1}=P^{T}$) where $\Delta$ is a diagonal matrix. Following the spectral properties of transition kernel (Lemme 9 in \cite{mixting}): if $K$ is ergodic then the maximal value in the diagonal $\Delta$ is $1$ and it is unique. The other diagonal elements have absolute value less than $1$. We can deduce the relation linking $W'_0$ and $C$ with the diagonal decomposition of $K$:
\begin{equation}
 %   C\, P^{-1} W W' P = \sum_{i=1}^C \Delta^i.
    P^{-1} W'_0 P = \frac{1}{C} \sum_{i=1}^C \Delta^i.
\end{equation}
The right side of equation can be calculated exactly by the formula of geometric sum: it is a diagonal matrix which terms are $\frac{1}{C} \sum_{i=1}^C \lambda^i$ % = \frac{1}{C} (\frac{1- \lambda^{C+1}}{1-\lambda} - 1)$
for all $\lambda \ne 1$.
\\
Now we can show the equivalence:
\begin{itemize}
    \item Given a stochastic matrix $K$, we can always find $W'_0$ such that $W'_0$ is also a stochastic matrix and satisfying (\ref{eq:MK}).
    
    The matrix constructed in this way $K:=P \Delta P^{-1}$ is stochastic. $v = (1,1,...,1)$ is the eigenvector of $W'_0$ of eigenvalue $1$ such that $W'_0 v =v$ and the first column of $P$ is $v$. So $P^{-1} v = (1,0, ..., 0)$ hence $P \Delta P^{-1} v = P (1,0, ..., 0) = v$. $K v = v$ means that each row of $K$ sums up to $1$ which prove that $K$ is stochastic.
    
    \item  Inversely, given a \emph{diagonalizable}  stochastic matrix $W'_0$ with eigenvalues in $[0,1]$, we can also find such a matrix $K$. In fact, if $W'_0 = P D P^{-1}$ where $D$ is a diagonal matrix with value between $0$ and $1$, we can find a diagonal matrix $\Delta$ such that $\frac{1}{C}\sum_{i=1}^C \Delta^i = D$ because the function 
    $$x\mapsto  \frac{1}{C}\sum_{i=1}^C x^i = \frac{1}{C}\left( \frac{1- x^{C+1}}{1-x} - 1\right)$$
    is bijective from $[0,1]$ to $[0,1]$. 
    % illustration: https://www.wolframalpha.com/input/?i=plot+1%2Fy+*+%28+%281-x%5E%28y%2B1%29%29%2F%281-x%29+-1%29%2C+x%3D0..1%2C+y%3D1..20
\end{itemize}
~\hfill$\square$

\subsection{Understanding Word2Vec as a Compression of a Reference Model}
\label{sec:compression}

\subsubsection{How to compress a Reference Model} 

 Assume that we start with the Reference Model $(\Id, W'_0,C)$ of a corpus and try to qualitatively understand what one can expect from a word2vec model of the same corpus when $N$, the dimension of the word embeddings, is $V-1$:
 \begin{itemize}
     \item On the one hand the likelihood associated by the word2vec model to the pair of words $(x_i, x_j)$ reads
     $$\softmax(x_i^T W W') x_j,$$
     \item on the other hand the empirical occurrences in the corpus says that this likelihood should be
     $$x_i^T W_0' x_j.$$
 \end{itemize}
If a pair of words $(x_i,x_{i'})$ is such that $x_i^T W_0' x_j = x_{i'}^T W_0' x_j$, then the rows $i$ and $i'$ of $W_0'$ are identical.
The best way to compress $W'_0$ is hence
\begin{enumerate}
    \item to use the embedding matrix to map rows $i$ and $i'$ on the same embedding.
    We will use the notation $E[i'\mapsto i]$ for this embedding matrix that is the $V\times V$ identity matrix with column $i'$ that is removed, and with $1$ in place of $0$ at its element $(i',i)$
    \item to remove the row $i'$ of $W'_0$. We use the notation $W_0'[i']^\ominus$ for a $(V-1)\times V$ matrix corresponding to $W_0'$ once its $i'$-th row is deleted.
\end{enumerate}
Then $E[i'\mapsto i]\cdot W'_0[i']^\ominus=W'_0$. As a consequence:
\begin{property}[Trivial compression of a Reference Model by a Word2vec]
With the upper notations: the word2vec model $(E[i'\mapsto i],W'_0[i']^\ominus,C)$
is an \emph{exact compression} of the Reference Model $(\Id,W_0',C)$, such that the $i$-th and $i'$-th words of $W_0'$ are identical, i.e.:
\begin{equation}\label{eq:trivial:compress}
   \forall (k,k'):\: \softmax(x_k^T \, E[i'\mapsto i]\, W'_0[i']^\ominus) x_{k'} = x_k^T W_0' x_{k'}.
\end{equation}
\end{property}
We can have a qualitative look at $E[i'\mapsto i]$, the matrix of embeddings of this exact compression, and notice that it mapped $i$ and $i'$ on the same embedding vector.
And, qualitatively once more, the nature of words $xi$ and $x_{i'}$ is that they have the same ``probability vector", i.e. the same probability of occurrence of surrounding words. See \cite{implicitmatrix} and references therein to different perspectives on compression.

\paragraph{A frequentist (and confusing) viewpoint on synonyms.}
From a frequentist viewpoint, $x_i$ and $x_{i'}$ are interchangeable, this is the reason why the training algorithm will map them on the same embeddings. One could call them ``{synonyms from a word2vec viewpoint}'' or \emph{frequentist synonyms}.
For two words being frequentist synonyms it is enough that in the corpus all sentences containing $x_i$ have a similar sentence where it is replaced by $x_{i'}$.
This is clearly not the definition of semantic synonyms, for instance if all sentences containing the word ``\emph{bad}'' have an copy with the word ``\emph{good}'' in place, like:\\[.2em]
\begin{tabular}{l|lcl}
\rule{2em}{0pt}&\emph{look at this \framebox{bad} guy}& $\leftrightarrow$ &\emph{look at this \framebox{good} guy}.\\
&\emph{this is \framebox{bad} English} & $\leftrightarrow$ &\emph{this is \framebox{good} English}.\\
&etc.
\end{tabular}\\[.2em]
and these semantic antonyms will become synonyms in embeddings learned on this corpus.
Elaborating a little more on application, one can bet that on a corpus of cars, colors will probably be frequentist synonyms, since most cars can be seen in any color, hence such a corpus may have as many sentences like: ``a \emph{gray} car had an accident'' than ``a \emph{blue} car had an accident''. But on a corpus made of cooking recipes, colors will not be frequentist synonyms because there is no ``\emph{red} bananas'' but ``\emph{red} apples'' and ``\emph{red} peppers''.
We will empirically see in Section \ref{sec:finance} that if frequentist synonyms provide structure to a word2vec mode, they can confuse a task related to exploit polarity of sentiments of financial texts.

\subsubsection{A formal definition for compression of Markov chains}

The notion of Reference Model enables to define clearly what kind of compression can be expected from a text embedding, in particular for the word2vec class of models.
The vocabulary size being $V$, and the embedding dimension being $N$, compressing \emph{linearly} a Reference Model $(\Id,W'_0,C)$ using a word2vec will end up with a ${V\times N}$ matrix $R$ and a ${N\times V}$ matrix $R'W_0'$, such that the probabilities defined by $(R,R'W_0,C)$ are as close as possible from the original ones.
Qualitatively:
$$\forall x_i,x_j:\: \softmax(x_i^T RR'W_0')x_j \mbox{ ``close to'' } x^T_i W_0' x_j.$$
The main point is to choose a quantitative criterion that has a sense to define this desired proximity.
Qualitatively, it is clear that:
\begin{itemize}
    \item the rank of $RR'$ needs to be maximal (taking into account that the more frequentist synonyms in the language the lower the rank of $W'_0$) to recover the space spanned by $W_0'$,
    \item without the softmax non linearity and if the SVD (Singular Value Decomposition) of $W'_0$ is $U\Sigma V^T$, one could expect that $R\simeq U\Sigma$ and $R'\simeq V^T$ (see \cite[Section 4.2]{implicitmatrix} for more details).
    Nevertheless the softmax changes the setting. Moreover, the SVD compression criterion is the minimization of the unexplained variance whereas in word embedding, the proximity of conditional probabilities seems to be a better criterion.
\end{itemize}

\begin{definition}[Compression of a reference model $({\rm Id},W_0',C)$]
We define the compression to dimension $N$ of a reference model of size $V$ using two 
mappings $\phi(W'_0)$ and $\phi'(W'_0)$.
%matrix  $R \in \mathbb{R}^{V\times N}, R'\in \mathbb{R}^{N\times V}$. 
The compressed model is the $(\phi(W'_0),\phi'(W'_0),C)$ word2vec model.
For simplicity of the notation we set 
$$\Phi(W'_0):=\phi(W'_0)\phi'(W'_0)$$
and we define $\frakC$ the class of function $\Phi$ that can be written this way.
\end{definition}

If we consider the sequence of words $X_k$ of a text of length $T$ as a stochastic process, then we may take any distance $d(\cdot,\cdot)$ between probability vectors (i.e. $x_k^tW'_0$ or $\softmax(x_k^T WW'$) and write our optimization problem as
\begin{equation}
\label{eq:op.mv}
    \min_{R, R'}\; \frac{1}{T} \sum_{k=1}^{T} d(\softmax(X_k^T \phi(W'_0)\phi'(W'_0)), \ X_k^T W'_0)
\end{equation}
For linear compression, just note that $RR'W_0':=\Phi(W'_0)$. 
A last qualitatively remark: $\softmax(X_k^T \Phi(W'_0))$ and $X_k^T W'_0$ are both row vectors which coordinate $j'$ represents the probability of occurrence of the $j'$-th word of the vocabulary (in the neighborhood of $X_k=x_i$, i.e. if the word $x_{j'}$ occurs at least once in the next $C$ words after $X_k=x_i$). Hence the former will be a ``good compression'' of the later if the these two probability distributions are ``close".
\\
A natural choice is to consider the \emph{cross entropy} (that is a shifted version of Kullback-Leiber divergence) between these two distributions:
\begin{eqnarray}\label{eq:CE:basic}
    \CE(\Pr_{W'_0|X_k}, \hat\Pr_{\Phi(W'_0)|X_k})&:=& -\Esp_{X_k^T W'_0}\log\softmax(X_k^T \Phi(W'_0))\\\nonumber
    &=& -\sum_{x_{j'}\in\calV} \log\softmax(X_k^T \Phi(W'_0))x_{j'} \cdot X_k^T W'_0 x_{j'} 
\end{eqnarray}
Thanks to the ergodic theorem\footnote{i.e. $\lim_{T\rightarrow\infty}\frac{1}{T}\sum_t f(X_k)=\Esp_\mu f(X_i)$.} that can be used on the stochastic process made of the sequence of words when the length of the corpus goes to infinity, it is now possible to state this definition of a compression:
\begin{definition}[Compression criterion on a Markov chain generated text]
If a text $\frakX=X_1,\ldots,X_k,\ldots$ stems from a Markov generative model which Reference Model is $(\Id, W_0',C)$, we define a compression $\Phi$, taken in the class of functions $\frakC$, of its Markov kernel $W_0'$ thanks to the expectation (according to the invariant distribution $\mu$ of $W_0'$) of the Cross Entropy between the output vectors of a word2vec. 
It reads:
\begin{equation}\label{eq:compress:markov}
  \min_{\phi\in\Phi}\; \Esp_{X_k\sim\mu} \,\CE\left(\Pr_{W'_0|X_k}, \softmax(X_k^T \Phi(W'_0))\right).
\end{equation}
\end{definition}
Note that with the already defined notation $\hat\Pr_{\Phi(W'_0)|X_k}$ for $\softmax(X_k^T \Phi(W'_0))$, we recover a minimization that is compatible with (\ref{eq:CE:basic}).

Moreover this definition goes beyond the Skip-gram word2vec embeddings, since the compression $\Phi$ is quite generic at this stage. 
It is possible to extend it definition using to the natural filtration $\frakX_k$ associated to $\frakX$ and replacing 
$\softmax(X_k^T \phi(W'_0))$ by $\softmax(X_k^T \phi(W'_0, \frakX_k))$ in formula (\ref{eq:compress:markov}). It would need to change the writing of the expectation but it allows the compression to use all the words up to the $k$-th word of the text, and hence to embed a memory or local metrics, addressing for instance part of the Weakness 1 exposed in Section \ref{sec:PLM}.

\subsubsection{Convergence of the Skip-gram word2vec loss function to a cross-entropy}

Show that the mean of Skip-gram loss function converges to the criterion when the size of the corpus size goes to infinity, i.e. replacing it by something link $\mathbb{E}_t \ell(X_t)$ that would be close to (\ref{eq:compress:markov}).
% Explain corss entropy + P hat there

\begin{theorem}[Correspondence between Skip-gram loss function and compression criterion]
\label{th:cross}
Assume the words $X_1,\ldots,X_T$ of a document of length $T$ are generated thanks to an ergodic stochastic process according to a Reference Model $(\Id,W'_0,1)$ and an initial probability distribution $m_0$ over the vocabulary.
Then the loss function of a Skip-gram word2vec model $(W,W',1)$ over this corpus converges towards the expectation of cross-entropy between $\Pr_{W'_0|X_k}$ and $\hat\Pr_{W,W'|X_k}$
\begin{equation}\label{eq:th:cross}
\lim_{T\rightarrow\infty} \frac{1}{T}\sum_{k=1}^{T} -\log \softmax(X_{k}^TW W')X_{k+1} = \Esp_{X_k\sim\mu} \CE(\Pr_{W'_0|X_k}, \hat\Pr_{W,W'|X_k}).
\end{equation}
\end{theorem}

\begin{proof}
For the sake of notations, we will restrict the theorem and its proof to $C=1$. % (hence $W'_0 = K$). 
Note $\mu$ as the 
ergodic measure of the stochastic process $(X_k)_k$,
% invariant measure of Markov chain $(X_k)_k$.
i.e. $\mu(x_i) = \lim_{T \to \infty}\frac{}{T} \sum_{k=1}^{T} \mathbbm{1}(X_{k}=x_i)$.
Then we can write the loss function
\begin{align*}
   -\frac{1}{T}&\sum_{k=1}^{T} \log \softmax(X_{k}^TW W')X_{k+1}\\ 
   =& -\sum_{i=1}^V\frac{\sum_{k=1}^{T} \mathbbm{1}(X_{k}=x_i)}{T}
   \sum_{j=1}^V \frac{\sum_{k=1}^{T} \mathbbm{1}(X_{k}=x_i, X_{k+1} = x_j)}
   {\sum_{k=1}^{T} \mathbbm{1}(X_{k}=x_i)}\log \softmax(x_{i}^T W W')x_j   
\end{align*}
With the notations $\Pr_{W_0'|x_i}$ for the probability distribution of words given the occurrence of $x_i$ (i.e. $x_i^T W'_0$) and $\hat\Pr_{WW'|x_i}$ for the probability of the same events modelled by the Skip-gram word2vec (i.e. $\softmax(x_i^T W W')$). 
Then limit when the number of words goes to infinity of the Skip-gram loss reads
\begin{align*}
   -\lim_{T\to\infty}& \frac{1}{T}\sum_{k=1}^{T} \log \softmax(X_{k}^TW W')X_{k+1} \\
   = & -\sum_{i=1}^V \mu(x_i) \sum_{j=1}^V \frac{\mu(x_i)x_i^T W'_0 x_j}{\mu(x_i)}\log \softmax(x_{i}^TW W')x_j\\
   =& \sum_{i=1}^V \mu(x_i) \; \CE(\Pr_{W_0'|x_i}, \softmax(x_i^T W W')) \\
   =& \Esp_{X_k\sim\mu} \CE(\Pr_{W_0'|X_k}, \hat\Pr_{W W'|X_k}). 
\end{align*}
\end{proof}
When $C>1$ the deduction is similar. 

\section{Synthetic experiments: empirical study of the role of structures}
\label{sec:synth:experiments}

In this section, we will leverage on different elements of our theoretical analysis to perform numerical explorations around the identifiability of word embedding models. We mainly leverage on these two elements
\begin{itemize}
    \item under restrictive assumption, we can use a Markov chain with a kernel $K$ to build a Reference Model $(\Id,W_0',C)$ thanks to equation (\ref{eq:MK});
    \item the existence of \emph{frequentist synonyms} influences the capability to compress efficiently a Reference Model.
\end{itemize}
We will hence generate different Reference Models having more or less structure (here structure means having blocks of frequentist synonyms or not), on vocabularies of different sizes, and observe how easy or difficult it is to recover them from a generated corpus. 
If the compression performed by a word2vec is not a principal component analysis (since it is more ``low rank" than ``low variance" driven, because of the nonlinearity introduced by the softmax and because the minimized criterion in a cross entropy with an unknown Reference Model via a sample of text), it can nevertheless be expected to find commonalities with the pitfalls identified long ago by Random Matrix Theory (see \cite{tao2012topics} for an overview) that are playing an important role in the theoretical understanding of the limits of deep learning \cite{choromanska2015loss}.

%In our numerical experimentation, we generate text corpus by a Markov chain $(X_k)_k$ of kernel $K$ and calculate the referent model matrix $W'_0$ by $K$. We then initialize a Skip-gram neural network and train our model with the generated data.

\paragraph{Experimental conditions.}
We generate different Reference Models on vocabularies of size $V$, that will be ``compressed" via embeddings of dimension $N$ (for illustration purposes we even consider some configurations for which $N>V$). We use a Markov chain to build the reference model so that it is ``fully random" (uniformly generated), or has a structure. This structure is made of ``blocks" which rows are identical to each other (that is the definition of frequentist synonyms).

We train a Skip-gram word2vec during a large number of epochs (enough to stabilize the learning) via a standard SGD (Stochastic Gradient Descent, generally Adam with a fixed rate of $10^{-4}$) implemented in pyTorch and running on Google Colab or AWS using {\sc cuda} acceleration.

\paragraph{Types of structure.}
We experiment different level of structure:
\begin{itemize}
    \item \emph{No structure} when $K$ is a dense matrix generated randomly and uniformly in the space of stochastic matrices each of its rows is sampled by the uniform distribution in the space of simplex using Dirichlet distribution with $\alpha=1$.
    \item \emph{Structure} when $K$ is made of blocks of duplicated rows, when the number of block varies, the overall size of the structure is always kept constant.
    For instance: when $V=1000$ we can take either $160$ blocks of $5$ rows or $40$ block of $20$ rows, because $150\times 5 = 40\times 20$.
\end{itemize}
The expected role of blocks is the following: the intrinsic dimension of a Reference Model made of $B$ blocks of size $S$ plus ``noise" on $V-B\times S$ components is between $B+1$ and $B+(V-B\times S)=V-B(S-1)$. it is expected to be closer to $B+1$ in a ``low overfitting'' configurations and closer to $V-B(S-1)$ when the model considers that it is equivalent to learn one row or $S$ similar rows. Our theoretical analysis suggests that the Skip-gram word2vec should be more in the first configuration that in the second, since attributing one vector of embedding to a block is far more rewarding (in terms of the loss function) than attributing it to one isolated row of $W_0'$.

\paragraph{Criteria to monitor.}
We focus on two criteria:
\begin{itemize}
    \item The loss function during the learning, to analyse the performance of the compression. It is interesting to note that the loss function (\ref{eq:loss:skipgram}) being compatible with the cross entropy (\ref{eq:compress:markov}), and since we are generating synthetic dataset, in some case one can really expect a full success of the ``compression" that is in these cases a simple identification of the generating model.
    \item The distance between the compressed vectors and the original ones. Once again it is a way to quantify the success of the identification of the properties of the generative model.
    
    That for, we focus on the (mean) cosine similarity between two words $x_i$ and $x_{i'}$ defined as the scalar product between the two probability vectors:
    $$\CS^P(x_i,x_{i'}):= \sum_{j=1}^V \softmax(x_i^T WW') x_j \cdot x_{i'}^T W_0' x_j.$$
    To formulate a criterion that is minimal when two groups $G_1$ and $G_2$ of words are close, we will use
    \begin{equation}
    \label{eq:crit}
        d_C(G_1,G_2):=1-\frac{1}{\#G_1\cdot \#G_2}\sum_{x_i\in G_1}\sum_{x_{i'}\in G_2} \CS^P(x_i,x_{i'}).
    \end{equation}
    This metric will allow us to understand if words belonging to the same group of frequentist synonyms are closer between themselves rather than to other words.
\end{itemize}

\subsection{First experimentation: low vocabulary size}

Figure \ref{fig:cosine50} shows that the Skip-gram word2vect has difficulties to find the generative $W'_0$ when the dimension of the word2vec is larger than the original dimension. This exhibits a clear identifiability issue.
Of course in general the size of the embeddings is lower than the vocabulary size, so for word2vec this configuration should never occur in practice. Nevertheless, not that the relative size of the embeddings vs. the one of the vocabulary is not low for embedding-driven language models like BERT.

\begin{figure}[!ht]
    \begin{center}
    \includegraphics[width=.48\linewidth]{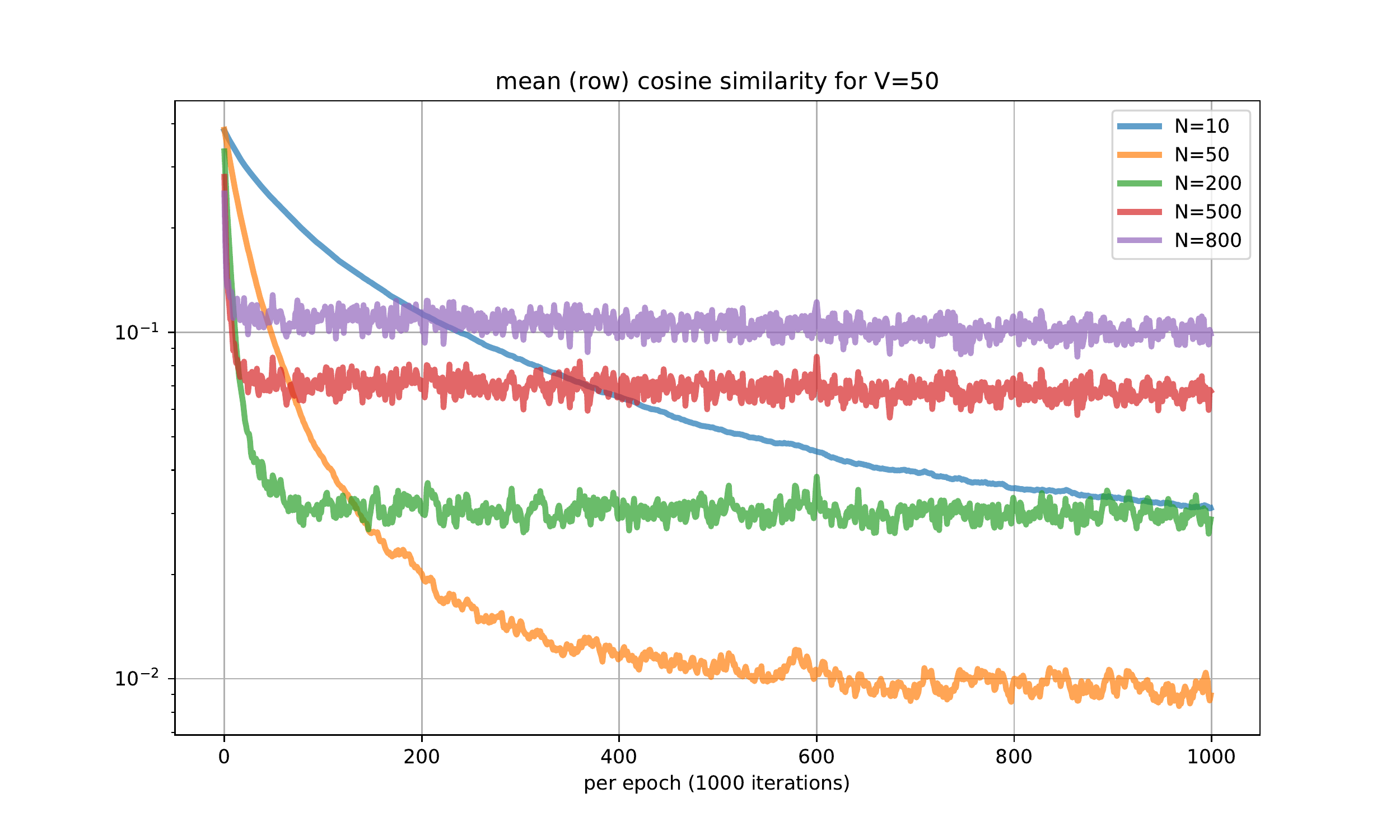}
    \includegraphics[width=.48\linewidth]{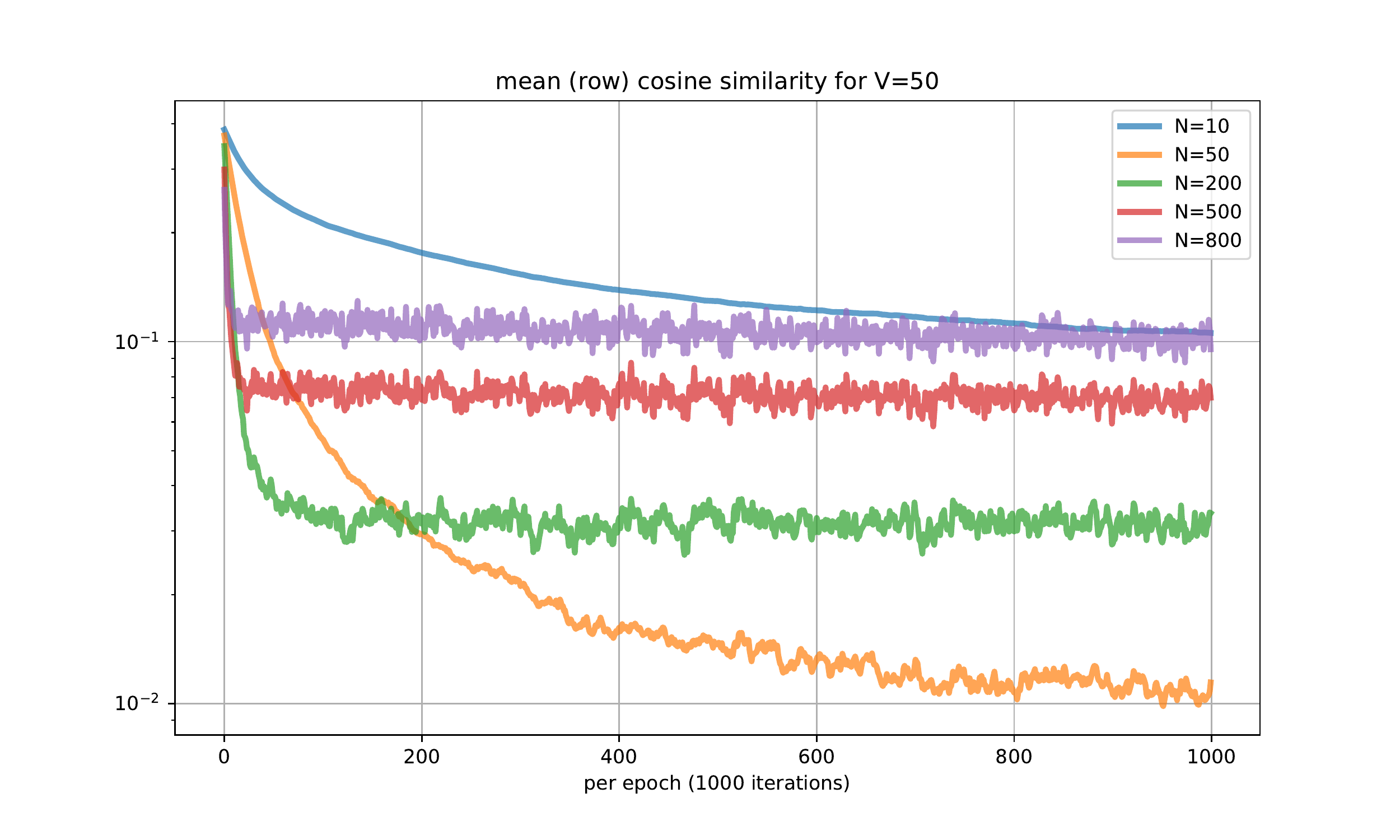}
    \end{center}
    \caption{The role of structure: cosine distance between $W_0'$ and $WW'$ during the learning ($x$-axis), for a vocabulary size $V=50$ and for 8 blocks of 5 (left) and no block (right).}
    \label{fig:cosine50}
\end{figure}

When the dimension of the embeddings is low, Figure \ref{fig:cosine200} suggests that the structure can be captured by the compression. 
Typically one could expect that an embedding of dimension $10$ has more chances to capture a Reference Model made of 8 blocks, than one made of 32 blocks, and it is verified in our experiments.

\begin{figure}[!ht]
    \begin{center}
    \includegraphics[width=.48\linewidth]{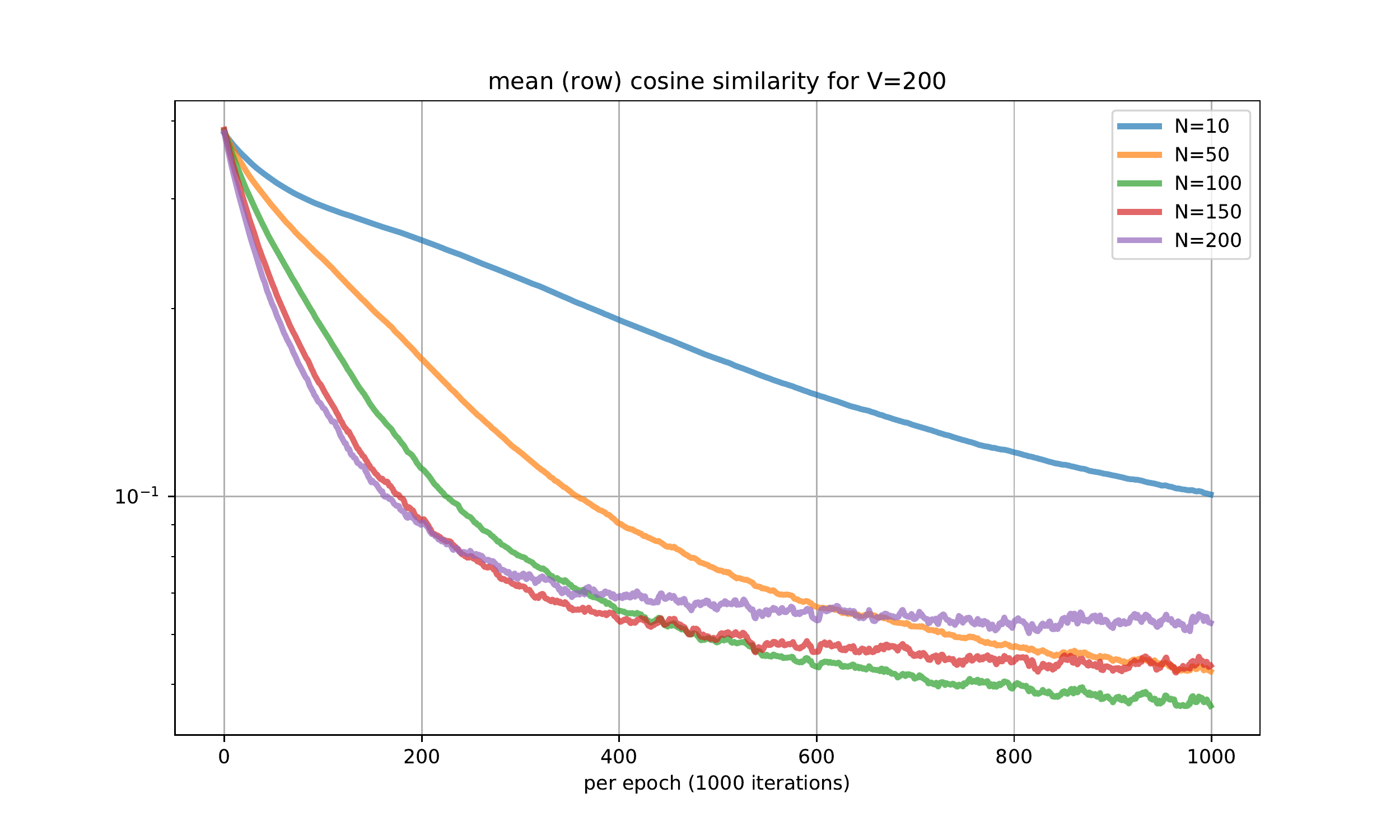}
    \includegraphics[width=.48\linewidth]{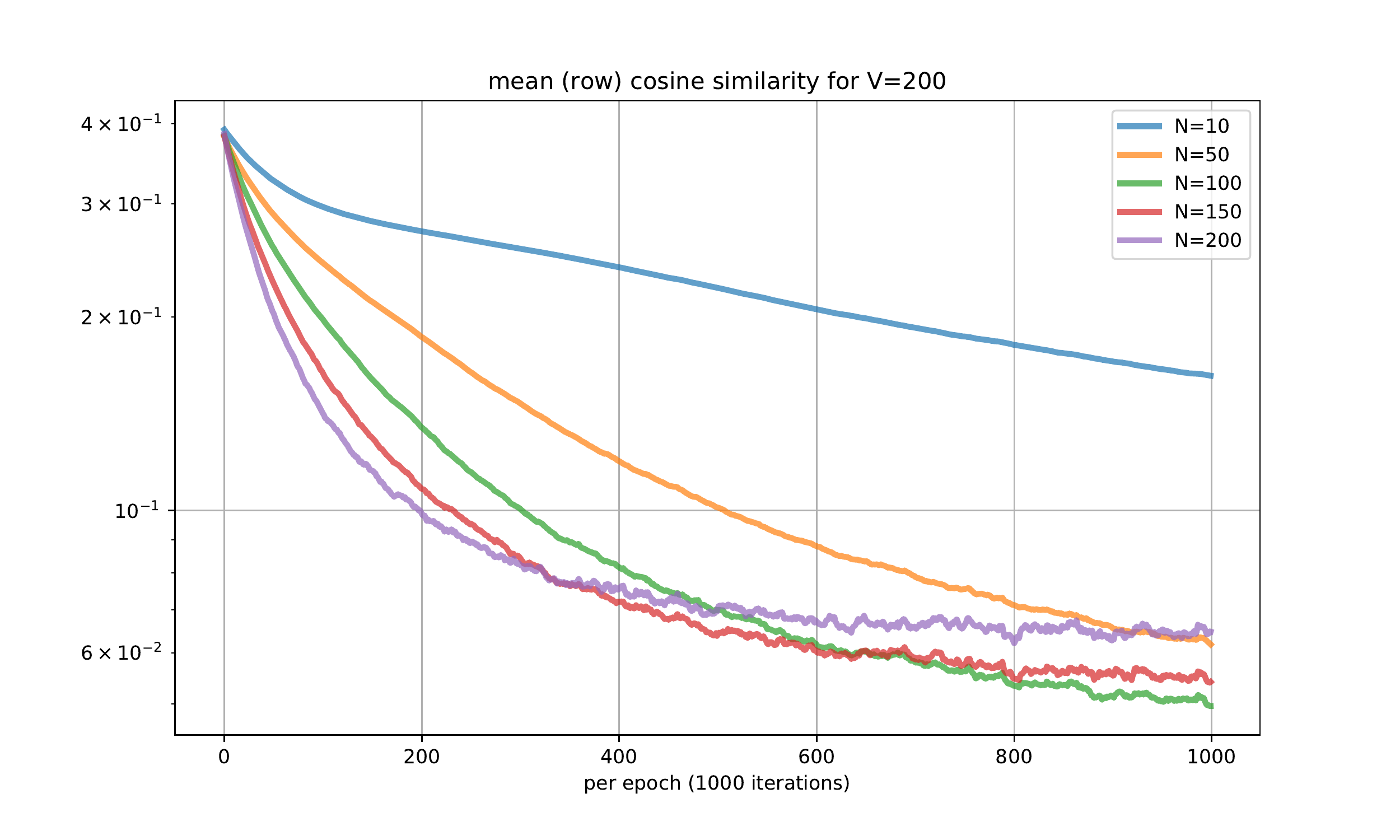}
    \\    
    \includegraphics[width=.48\linewidth]{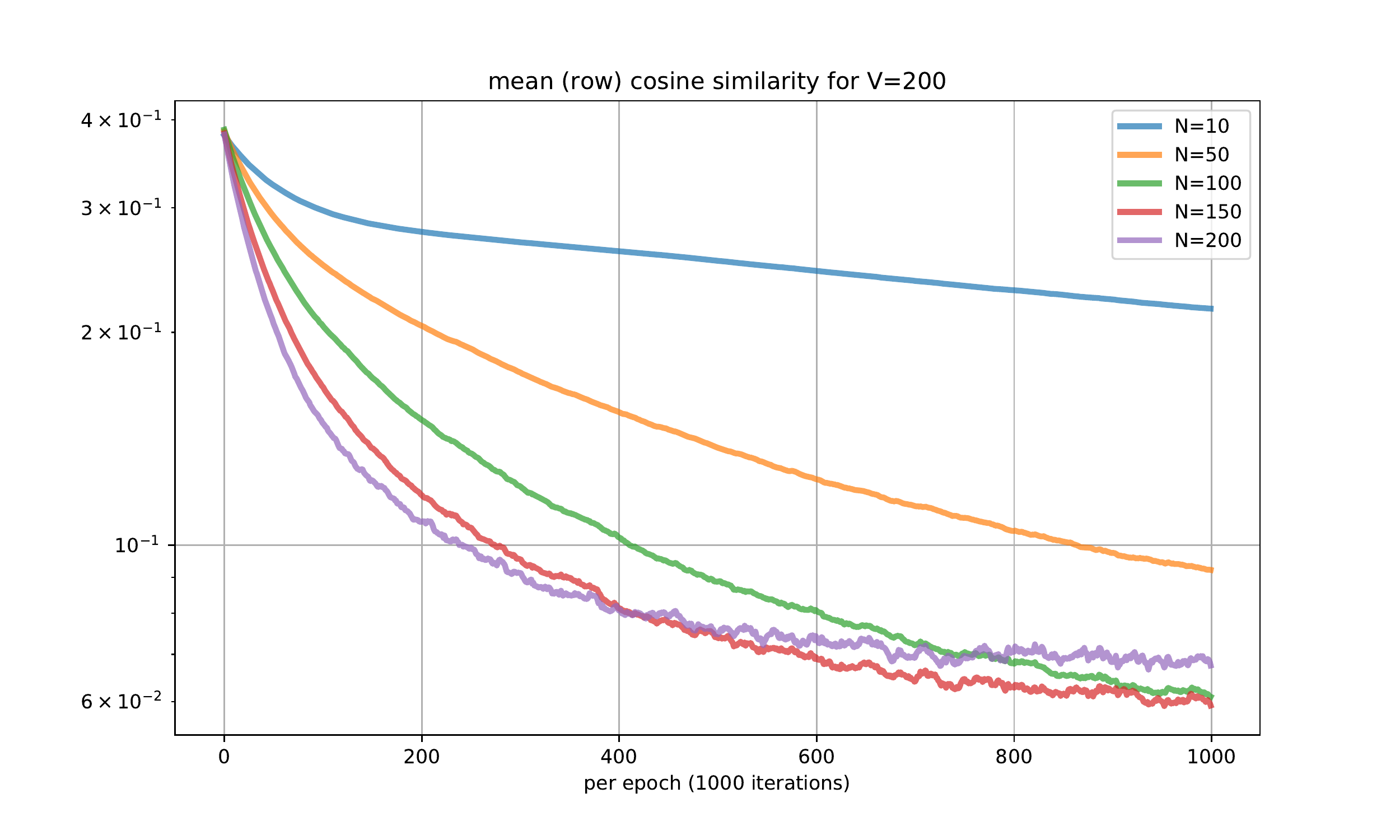}
    \includegraphics[width=.48\linewidth]{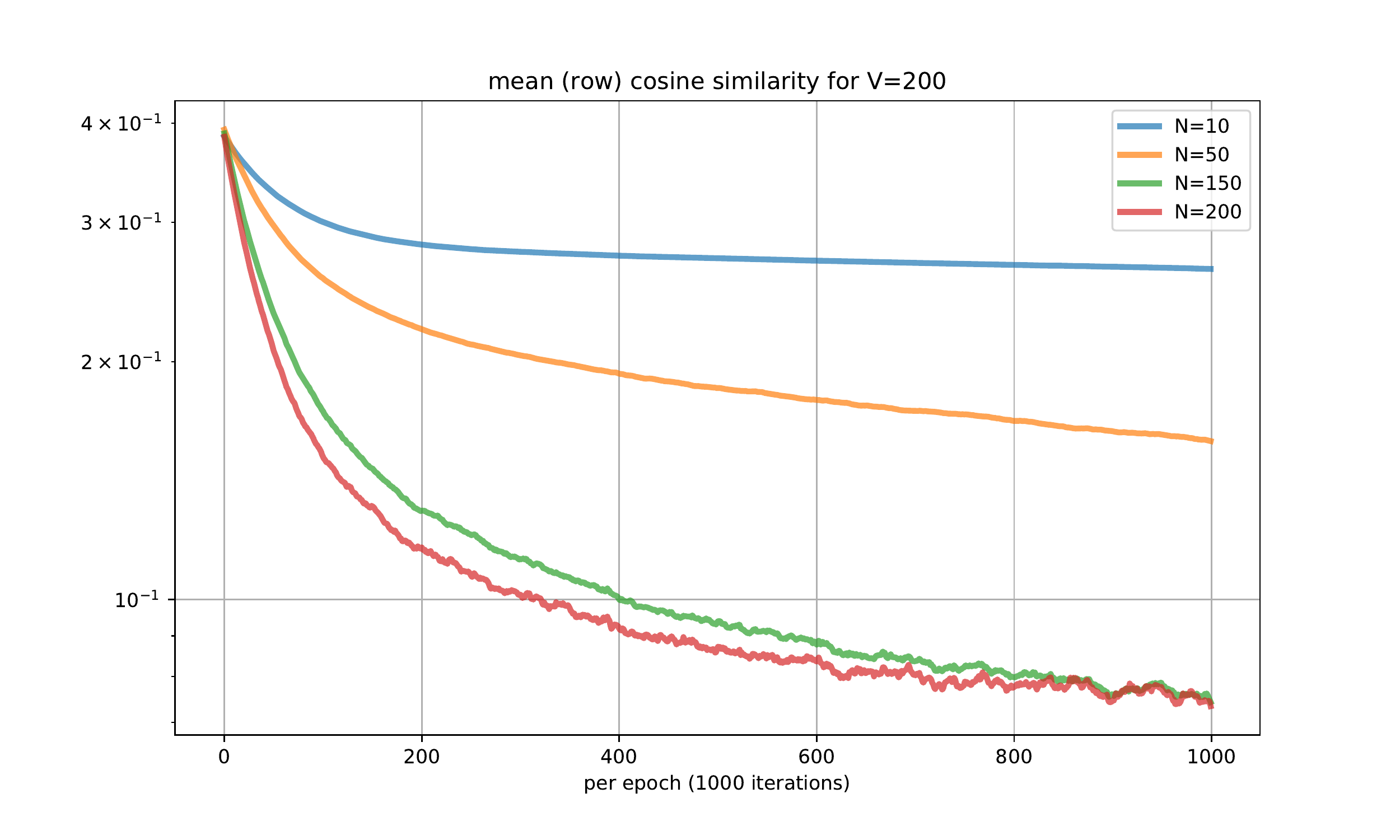} 
    \end{center}
    \caption{The role of structure: cosine distance between $W_0'$ and $WW'$ during the learning ($x$-axis), for a vocabulary size $V=200$ and for 8 blocks of 20 (top left), 16 blocks of 10 (top right), 32 blocks of 5 (bottom left) and no block (bottom right).}
    \label{fig:cosine200}
\end{figure}

In any case, Figure \ref{fig:cosine1000} underlines the fact that higher dimension of embeddings, even if it is lower than the vocabulary size, is better in presence of structure. For instance when the Reference Model is made of 160 blocks of 5 rows (i.e. 800 rows over 1000 of the Reference Model are exhibiting structure): embedding size of $200$ and $500$ perform similarly (with more noise for $N=500$) whereas $n=800$ performs very poorly. 
This underlines the fact that the structure has to be taken into account to choose the targeted compression dimension; in a highly structure language (with a lot of frequentists synonyms), a low dimension can be a proper choice. We can expect that very repetitive sentences (from a semantic perspective), like the headlines of financial News, should exhibit identifiability issues compared to the body of the same news, that are made of more semantically diverse sentences.

\begin{figure}[!ht]
    \begin{center}
    \includegraphics[width=.48\linewidth]{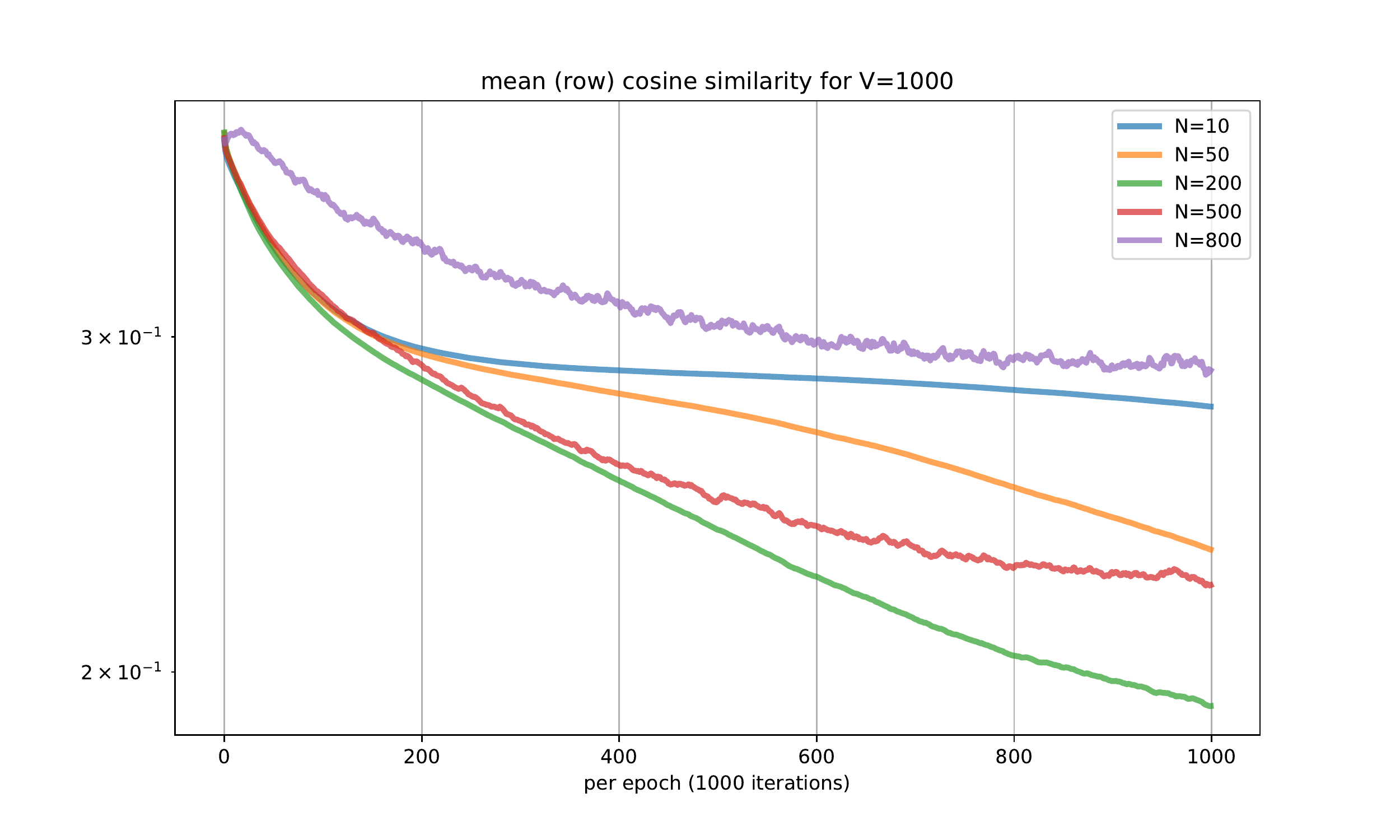}
    \includegraphics[width=.48\linewidth]{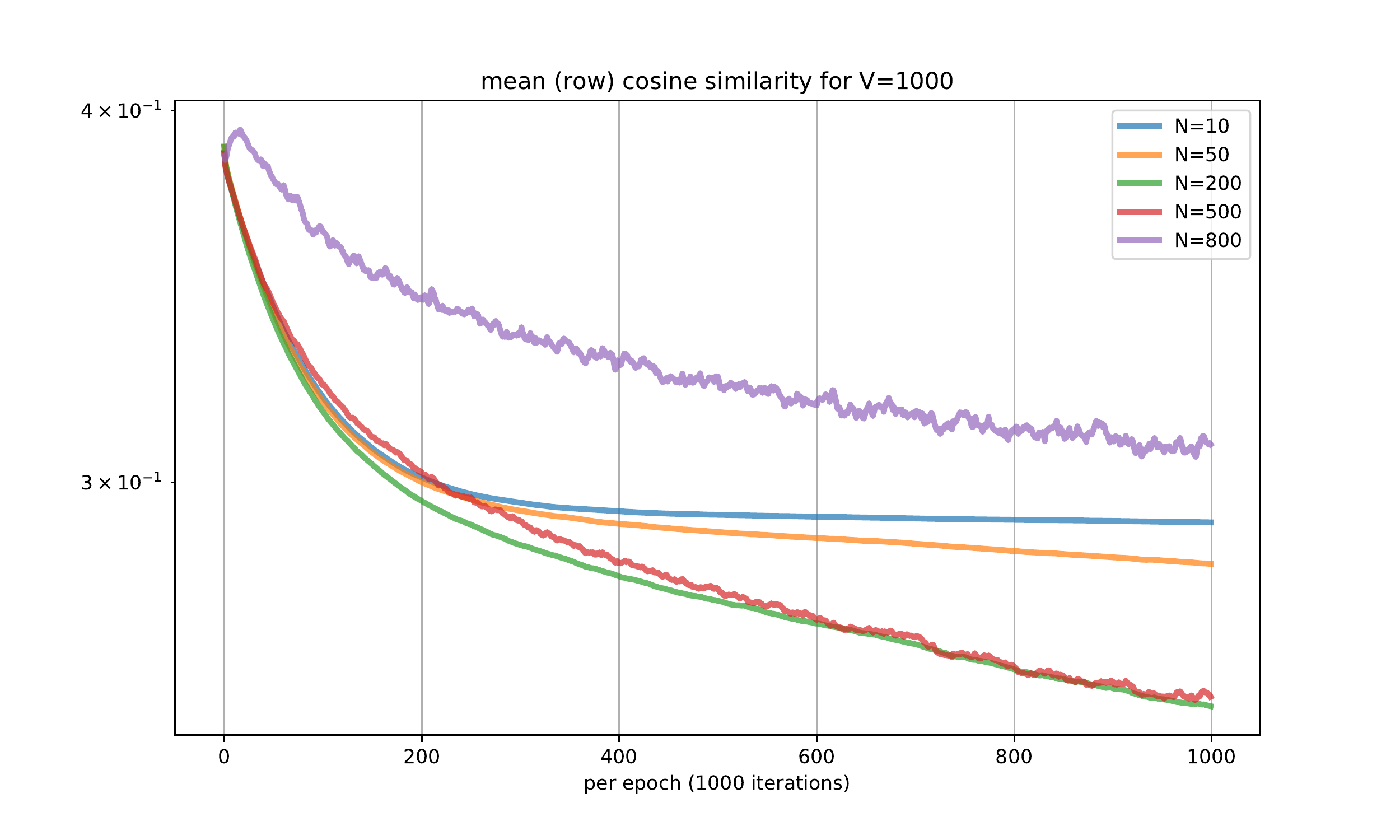}
    \\    
    \includegraphics[width=.48\linewidth]{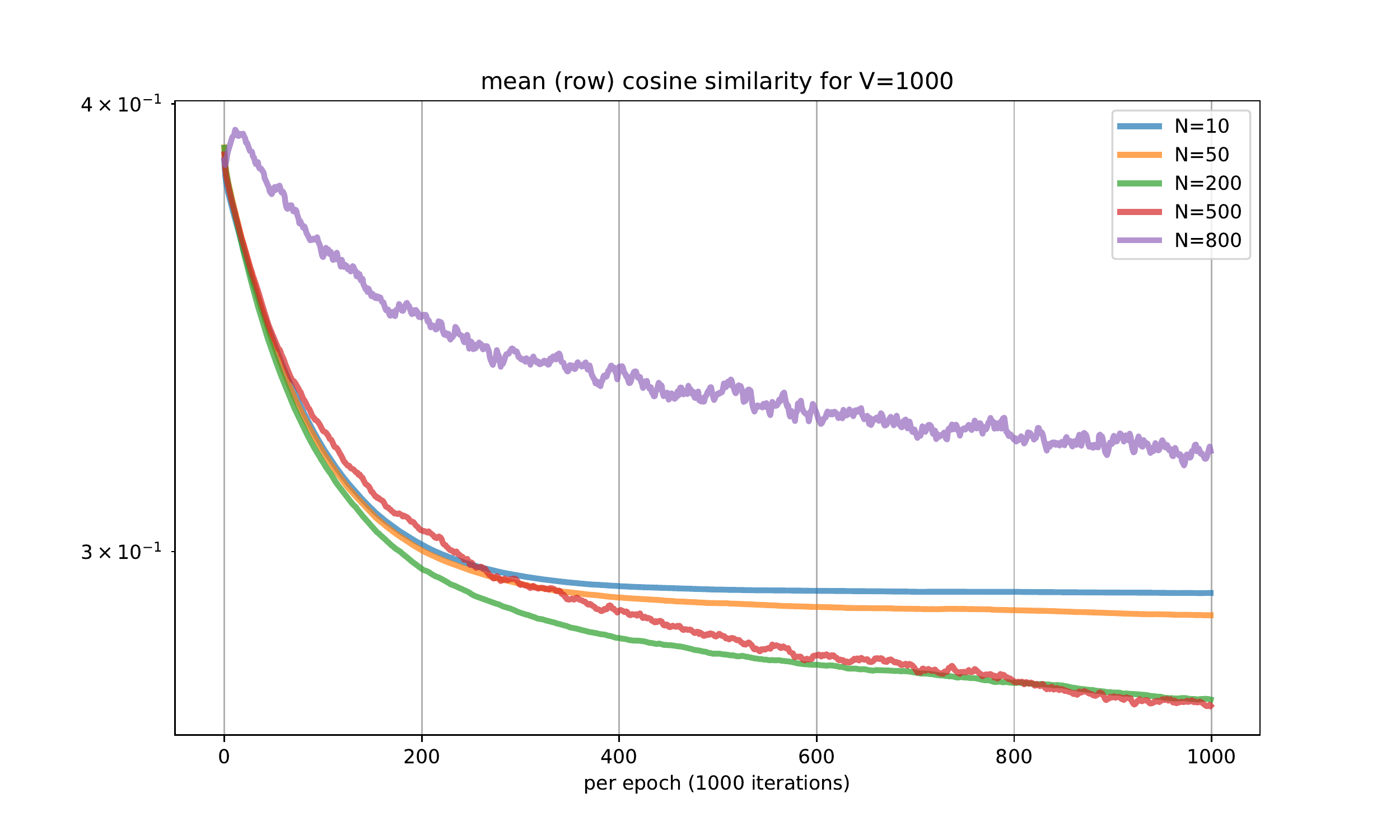}
    \includegraphics[width=.48\linewidth]{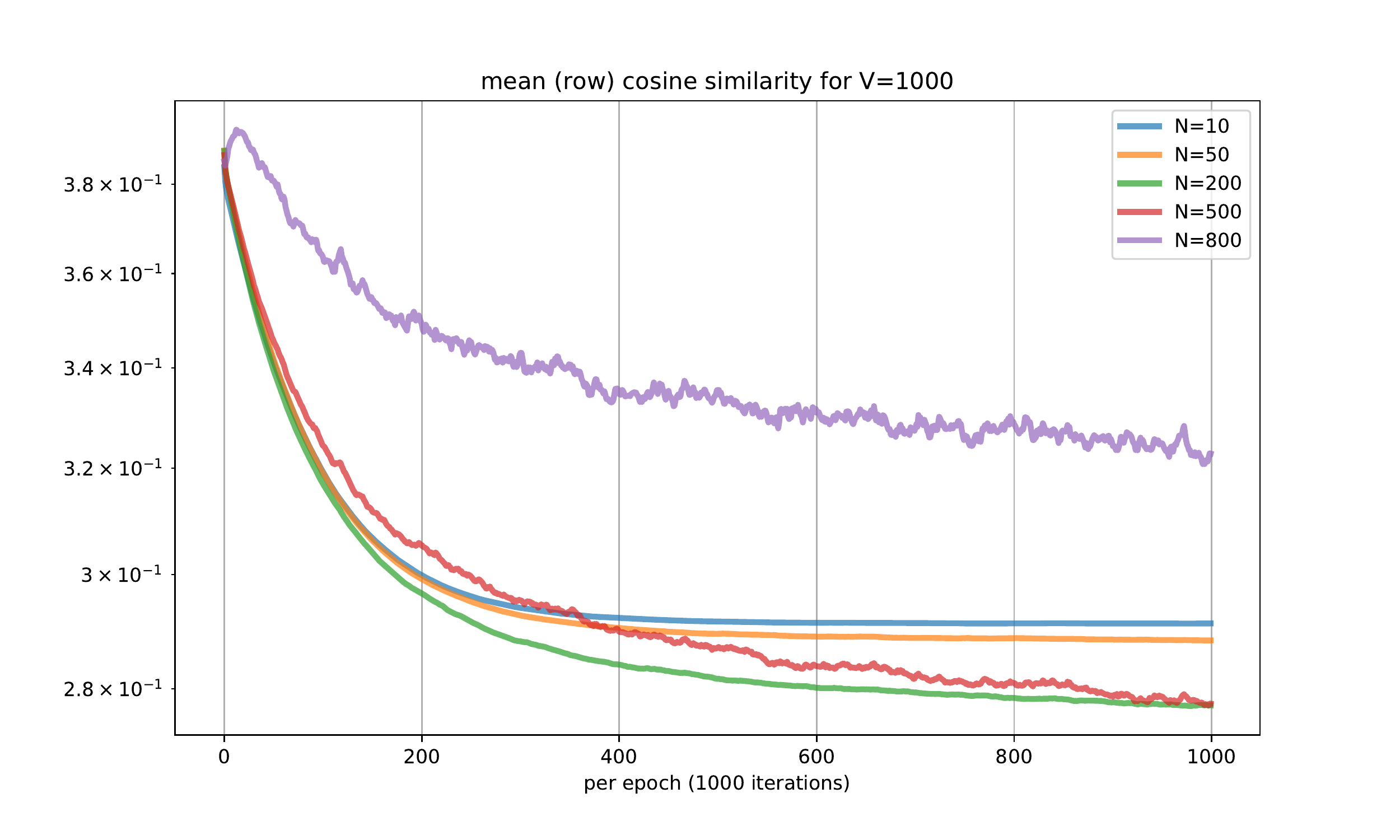}
    \end{center}
    \caption{The role of structure: cosine distance between $W_0'$ and $WW'$ during the learning ($x$-axis), for a vocabulary size $V=1000$ and for 10 blocks of 80 (top left), 40 blocks of 20 (top right), 160 blocks of 5 (bottom left) and no block (bottom right).}
    \label{fig:cosine1000}
\end{figure}

\subsection{Beyond the loss function: assessing the quality of the captured structure}

\begin{figure}[!ht]
    \centerline{\emph{8 block of size 5 for a vocabulary size of 50.}}
    \includegraphics[width=\linewidth]{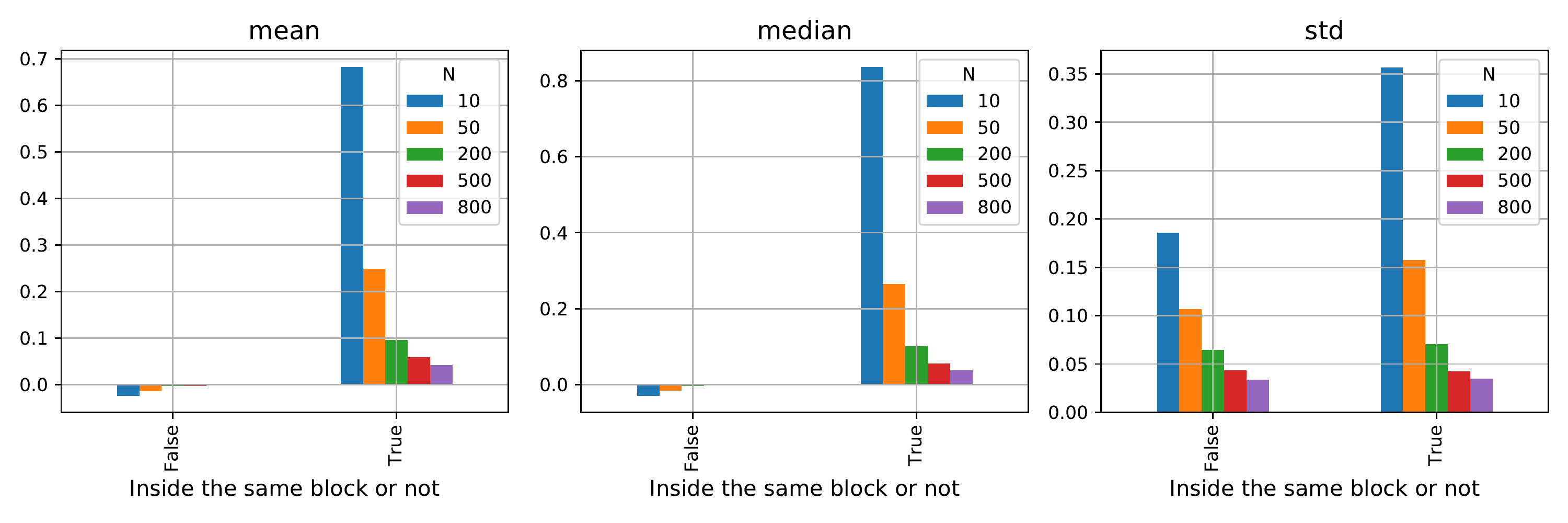}\\
    \centerline{\emph{40 blocks of size 20 for a vocabulary size of 1000.}}\\
    \includegraphics[width=\linewidth]{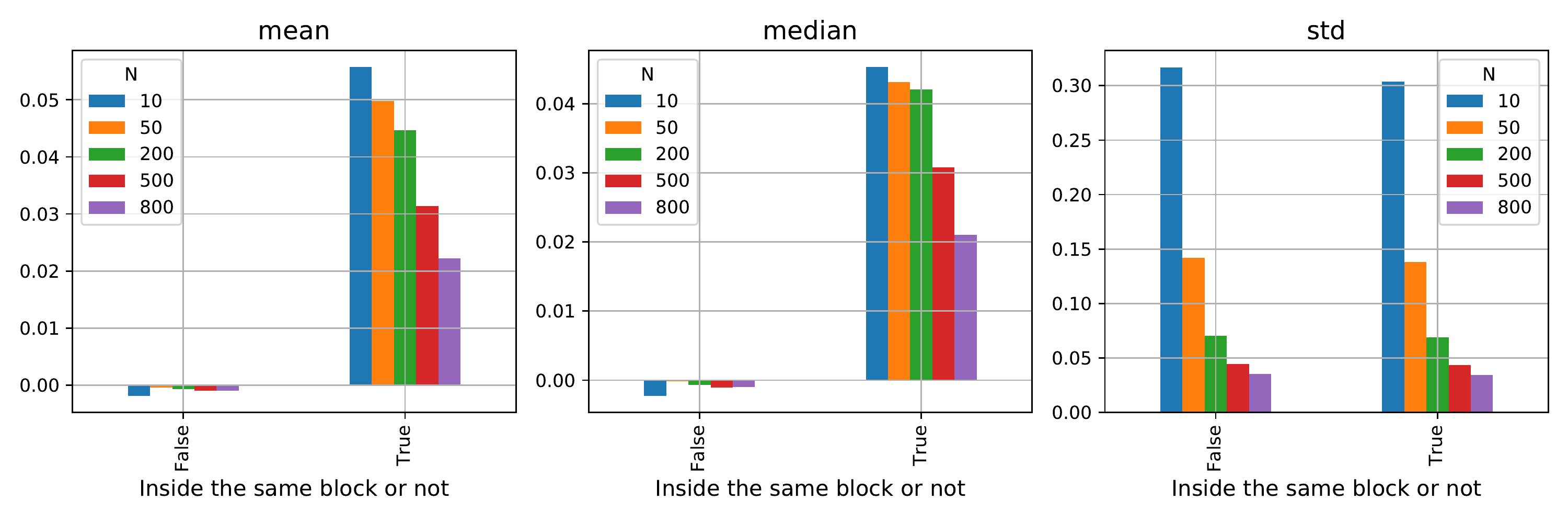}\\
    \centerline{\emph{160 blocks of size 5 for a vocabulary size of 1000.}}\\
    \includegraphics[width=\linewidth]{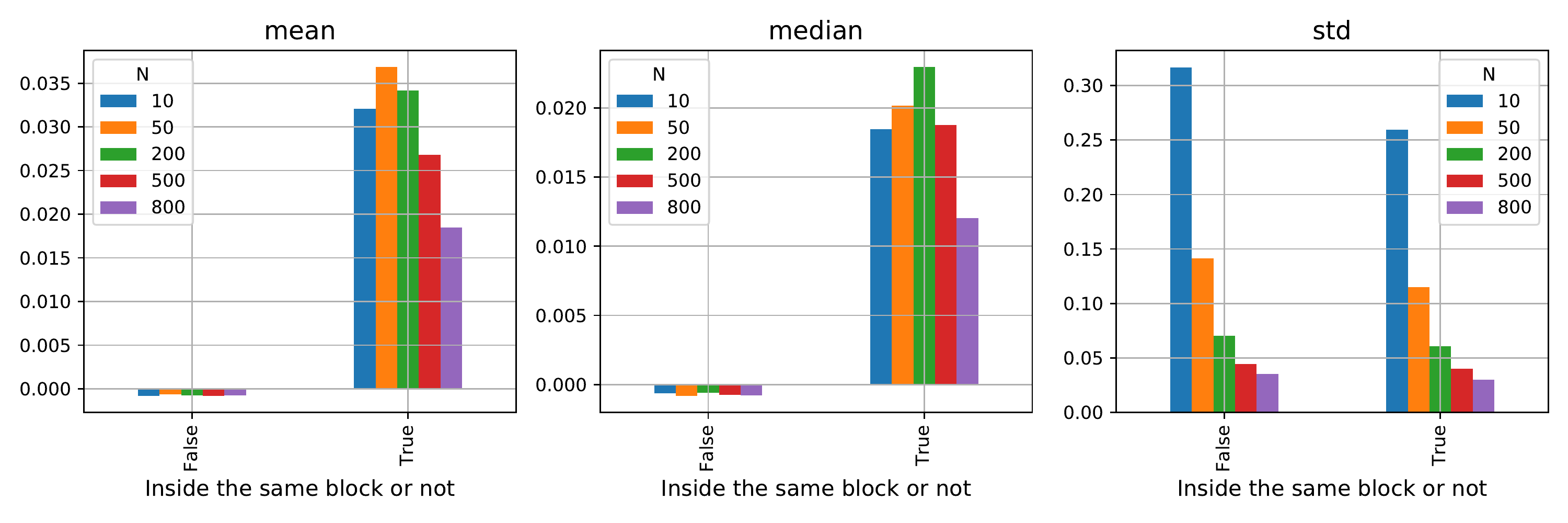}
    \caption{The recovery of structure: inter-block and intra-block cosine similarities of embeddings.}
    \label{fig:synthetic:blocks}
\end{figure}

Figure \ref{fig:synthetic:blocks} exhibits how the compressed model $WW'$ recovers the blocks of the Reference Model $W'_0$. The standard deviations inside a block or between words of another block are similar, and as expected the cosine similarity between words belonging to the same block is higher than the cosine between words of two different blocks. Nevertheless if this cosine is large in low dimension (0.8 in the best case of $V=50$) it is far lower in high dimension. It is 0.02 when $V=1000$ for $N\leq 500$, but it is halves when $N=800$, showing that overfitting is really present in such a case.

The qualitative conclusions to this empirical study for synthetic data are
\begin{itemize}
    \item the dimension of the embedding has to cope with the structure of the Reference Model, with a highly structured Reference Model, a too large embedding dimension is detrimental to the performances;
    \item the skip-gram word2vec faces identifiability issues in general, and we suspect it is not specific to this model. 
    \item Surprisingly, even if the performances of the model is not good, it succeeds in putting words closer to frequentist synonyms (i.e. within the same block) than to other words. 
\end{itemize}
Our experiments on the dimensionality of word embeddings seems to be compatible with \cite{dim} that is in favor of a dimension close to $N=300$. 

\section{Experiments on a corpus of financial news: What can be learned?}
\label{sec:finance}

This section is an empirical study learning embeddings on a financial corpus. Our goal is to explore how particular type of structure, that is usually named \emph{sentiment polarity} in NLP, is preserved or not by embeddings.

The usage of NLP in finance often targets to make the difference between ``good news'' and ``bad news'' on listed company, as a pre-requisit to build investment strategies (see \cite{gentzkow2019text} and \cite{xing2018natural} for details).
The ``Loughran-McDonlad Lexicon'' tuned by human experts on 10-K regulatory fillings by US corporates (see \cite{loughran2011liability} for details), is commonly used to quantify the polarity of text. 
Since it lists positive, negative, litigious words, we will use these lists of words as blocks of synonyms or antonyms. They study play a similar role as the ``blocks'' of the generative models studied in the previous section.

\subsection{Sentiments as a source of structure: Does Financial News understand better Loughran-McDonald's polarity than Wikipedia?}

\begin{table}[!h]\small
    \centering
\begin{tabular}{lrrrrrr}
\toprule
 & \multicolumn{3}{c}{News full text} & \multicolumn{3}{c}{News headlines} \\
 & avg. LM & avg. size &       nbe &   avg. LM & avg. size &       nbe \\
\midrule
2008 &    0.26 &    156.18 &   415,524 &      0.24 &     10.14 &   406,315 \\
2009 &    0.25 &    146.84 & 2,780,525 &      0.23 &     10.20 & 2,497,544 \\
2010 &    0.22 &    173.07 & 2,911,113 &      0.21 &     10.17 & 2,651,780 \\
2011 &    0.20 &    354.10 & 4,038,762 &      0.19 &      9.65 & 1,937,799 \\
2012 &    0.19 &    387.70 & 4,650,571 &      0.23 &     10.23 &   132,126 \\
2013 &    0.19 &    370.67 & 4,797,842 &      0.21 &      9.96 &   187,842 \\
2014 &    0.19 &    360.51 & 4,785,053 &      0.18 &      9.90 &   195,455 \\
2015 &    0.20 &    374.79 & 4,763,002 &      0.16 &      9.12 & 1,303,531 \\
2016 &    0.24 &    242.99 & 3,716,652 &      0.18 &      9.54 & 2,763,163 \\
2017 &    0.22 &    133.33 & 3,401,207 &      0.15 &      9.89 & 2,191,320 \\
2018 &    0.21 &    129.60 & 3,357,355 &      0.14 &      9.75 & 1,980,199 \\
2019 &    0.17 &    160.64 & 1,634,809 &      0.15 &     10.00 & 1,545,239 \\
2020 &    0.17 &    175.42 & 1,388,059 &      0.15 &      9.97 & 1,449,560 \\
\bottomrule
\end{tabular}
\caption{Descriptive statistics of the processed News full text
and Headlines (\emph{avg. LM} is the average number of words from the Longhran-McDonald lexicon in a document of the corpus
%Ce sont le nombre moyen des mots de LM dans un news/headline. Ce ne sont pas que des mots positives/négatives. Le news contient toujours un headline, c'est pour ça le gauche est > le droit, pourtant ce qui est surprenant est que les mots de LM sont plutôt dans les headlines car ils n'ont pas un très grand écart. Le formule de calcul: je comptes le nombres de LM dans chaque news, puis je fais le moyen de ce nombre.
, \emph{avg. size} is the average number of words, and \emph{nbe} is the number of news).}
    \label{tab:desc:bloom}
\end{table}

We use a corpus of financial News provided by a large provider of professional financial News covering years from 2008 to 2020. We will do experiments on the headlines only or on the whole body of the News (i.e. the headline followed by the text). We used different ad hoc filters to prevent the repetition of News (sometimes a News is repeated, in such a case we only keep its first appearance, and some times a header does not have enough words, in such a case we remove it from the headline corpus).
Table \ref{tab:desc:bloom} provides descriptive statistics on these corpora; it is restricted to the News that will be processed by our word2vec: the headlines contain on average 10 words where the body of the news contains on average around 250 words. Our filtering reduced the number of headlines years 2008, 2012, 2013 and 2014. Because of this we will restrict some of our analysis to other years only.

As a reference, we use embedding pre-trained in Wikipedia (see Appendix \ref{sec:tech:appendix} for details).

The content of the Loughran-McDonald Sentiment Word Lists (2018) is described in Table \ref{tab:LM:desc}. This lexicon contains more negative words than positive words; that is the two categories we will mainly focus on. Not all these words appear in the considered corpora: for instance embeddings trained on Wikipedia contains only 349 of the 353 positive words where the headlines of our financial News from 2008 to 2012 contains 343 of them.

Here are examples of the first words of this lexicon:
\begin{itemize}
\item Negative: \emph{abandon, 
abandoned,
abandoning,
abandonment,
abandonments,
abandons,
abdicated,
abdicates,} etc.
    \item Positive; \emph{able,
abundance,
abundant,
acclaimed,
accomplish,
accomplished,
accomplishes,
accomplishing,
accomplishment,} etc.
\end{itemize}

\begin{table}[!ht]
    \centering
    \begin{tabular}{c|r|rrr}
                   &    LM    & in wiki-& in Fin. News & in Fin. News\\
                   &    (ref) & pedia & headlines & Full text\\\hline % avg(BNB)?
Negative words     & 2,354 &     2,097 & 2,190 & 2,335\\
Positive words     &   353 &       349 &   343 & 354\\\hline\hline
Uncertainty words  &   396 &       270 &   273 & 294\\
Litigious words    &   903 &       590 &   637 & 820\\
StrongModal words  &    19 &        18 &    19 & 19\\
WeakModal words    &    27 &        26 &    26 & 27\\
Constraining words &   183 &       172 &   178 & 181\\ \hline
    \end{tabular}
    \caption{Number of word in each category present in each corpus.} 
    \label{tab:LM:desc}
\end{table}

Figure \ref{fig:avg:cosin:glob} exhibits the cosine similarity between different sections of polarized words (as it is a scalar product between two vectors of norm 1, its maximum value is 1). The exact computation process is the following: (1) we trained Skip-gram word2vec embedding on each year, (2) for each year we compute the cosine similarity between any two words of the considered lists (for instance between any two positive words, or between any positive word and any negative word), (3) we average over all the obtained cosine similarities (we plot the obatined time series in Figure \ref{fig:yearly:cosine}), (4) we average the obtained numbers excluding years  2008, 2012, 2013 and 2014 (Figure \ref{fig:yearly:cosine} shows that average cosine similarity is too different over these years for embeddings learned on headlines, probably because of the low number of headlines available).
\\
If the learned embeddings would have been compatible with the polarities expressed by the lexicon, we should observe
\begin{itemize}
    \item The highest average similarity for groups \emph{Positive-Positive} and \emph{Negative-Negative} (blue and orange bars),
    \item The lowest average similarity for the \emph{Positive-Negative} group (violet),
    \item Medium similarities for groups \emph{Positive-Random} and \emph{Negative-Random} (red and green).
\end{itemize}
\begin{figure}[!h]
    \centering
    \includegraphics[width=.8\linewidth]{images/compare_mean_polarities_T_selected.pdf}%{avg_cosin_group_global.pdf}
    \caption{Average cosine similarity between embeddings of different groups of polarities trained on financial News or on Wikipedia (\emph{Random} means either random words --for Financial News-- either all the other words --for Wikipedia--).}
    \label{fig:avg:cosin:glob}
\end{figure}
\emph{It is not what we observe}:
\begin{enumerate}
    \item if Positive-Positive and Negative-Negative (blue and orange bars) are the highest, it is not by far, especially for embeddings learned on the headlines.
    \item \emph{Positive-Negative similarity (violet) is never the lowest}, it is comparable to distance to random words for embeddings
    \item Positive-Random and Negative-Random (red and green) are not very low \emph{except for embeddings learned on the full text} of the News.
\end{enumerate}

\paragraph{Conjectures on frequentist synonyms amongst polarized worlds in finance.}
The first empirical conclusion we can make at this stage is that \emph{it is difficult for embeddings to separate financial polarized antonyms}.
The fact that the Positive-Negative (violet) bars are never lower than the Positive or Negative vs. Random (red and green) bars is an evidence of this difficulty. Nevertheless, \emph{it is easier to separate financial polarized words with embeddings learned on the full text of News rather than on headlines only}. We can conjecture that it is because financial headlines are written to be quickly understood by humans, and that for they contain similar sentences with positive or negative words in the same environments. Indeed: when a corpus is structure with short sentences of similar structure, the chances that semantic antonyms becomes frequentist synonyms are high. This configuration disappear when the corpus of full text (body) of financial News is used.
Last but not least Positive-Random and Negative-Random similarity is particularly close to zero when embeddings are learned on the full body of the News. This similarity is even lower for embeddings trained on the full text of financial News rather than on the Wikipedia corpus.

Our theoretical analysis shed light on this: it is difficult for the Skip-gram word2vec to make the difference between frequentist synonyms. From a corpus of short and very structured sentences, antonyms have good chances to become frequentist synonyms.

As a conclusion of this first analysis, we can conclude that using embeddings to discriminate polarity will be more difficult on a the corpus of headlines, that is not diverse enough, and financial enough (since it does not work is well on Wikipedia).

\subsection{Stationarity of embeddings associated to polarized words}

\subsubsection{The cosine similarity between polarities is stationary}

\begin{figure}[!h]
    \centering
    \includegraphics[width=.985\linewidth]{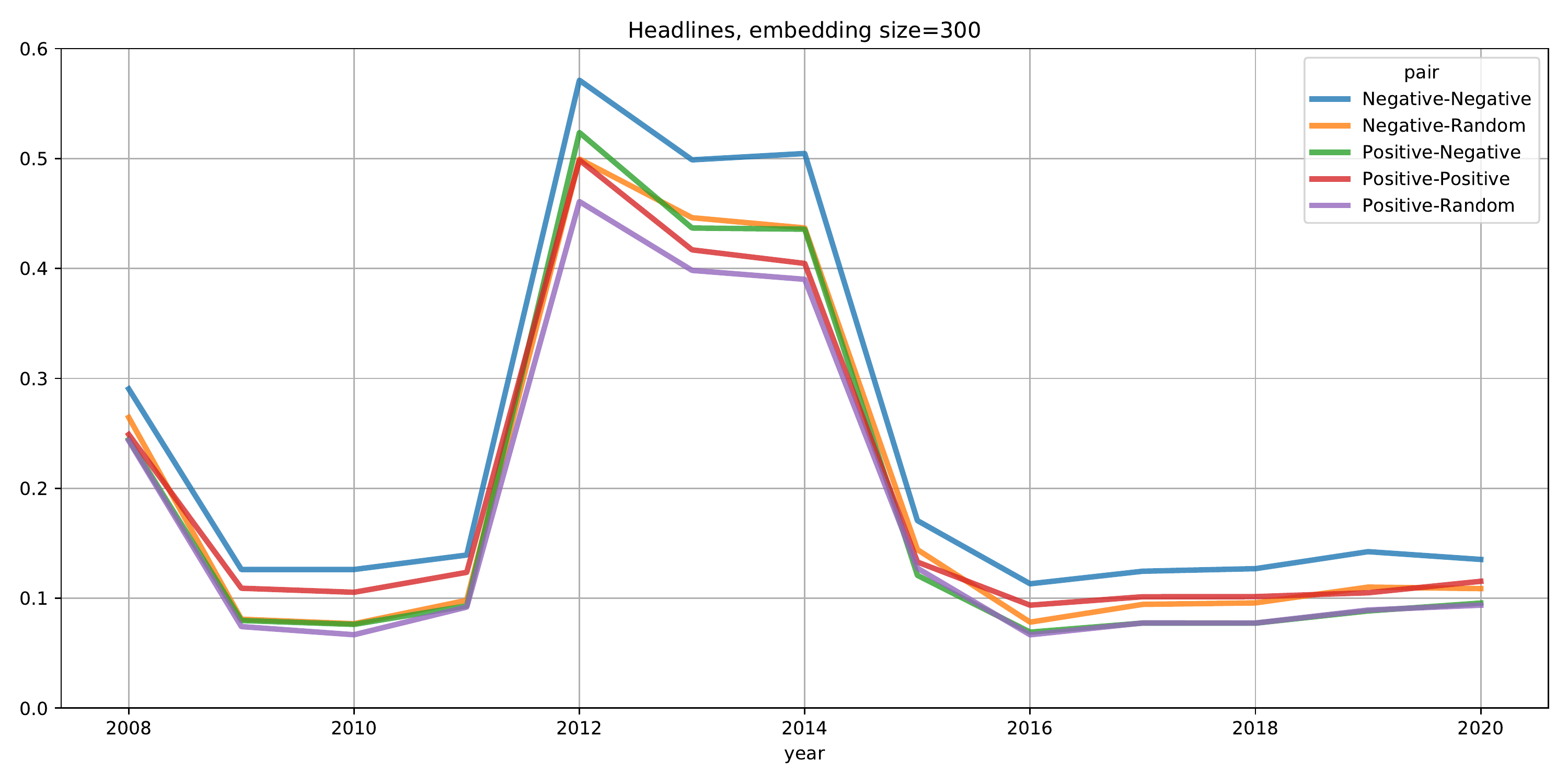}\\
    \includegraphics[width=\linewidth]{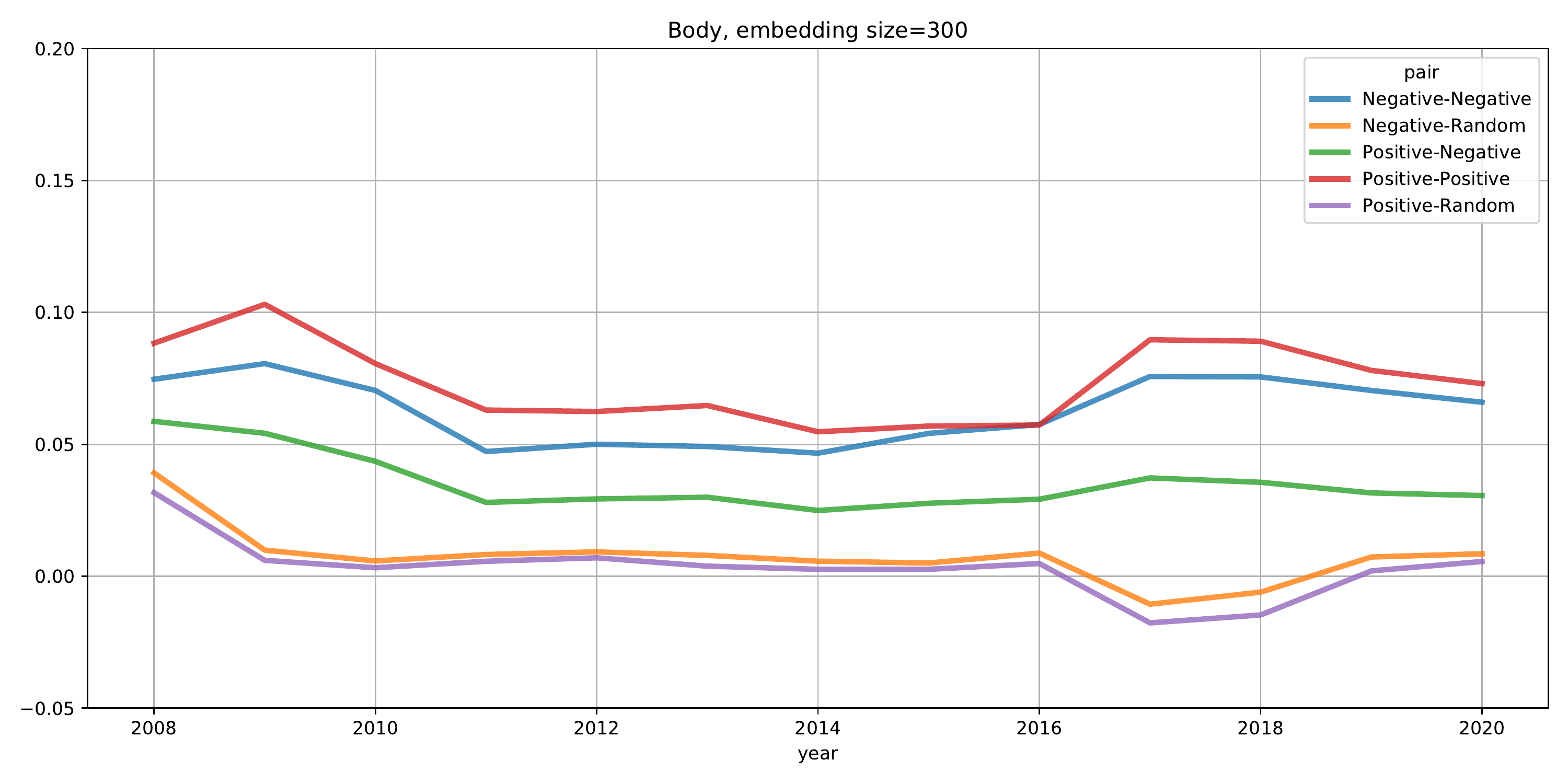}
    %{avg_cosin_group_yearly.pdf}
    \caption{Yearly cosine similarity between the embeddings of Loughran-McDonald Positive and Negative lexicon, for a skip-gram word2vec model learned on the headlines only (bottom) or on the full text (top) of financial News.}
    \label{fig:yearly:cosine}
\end{figure}

Table \ref{tab:desc:bloom} shows the number of document in our corpus each year; we have clearly less headlines in 2008, 2012, 2013 and 2014, and it is reflected in the average cosine similarities of Figure \ref{fig:yearly:cosine} (top panel). There is not enough documents, hence not enough diversity in the sentences, and as a consequence words are not been seen in differentiating enough contexts: the global cosine similarity is higher.

That being put aside, we observe that the averaged results of Figure \ref{fig:avg:cosin:glob} are in line with their time series representation: the ordering of cosine similarities between groups of words of same polarity are the same every year. Despite a widening of the difference between words of the same polarity (red and blue lines) and the cosine between polarized words and random words (orange and violet lines) in 2017 and 2018, \emph{the embeddings seem to be stationary} in the sense that there is no real change in their relative ranking from one year to another. 

\subsubsection{Influence of the embedding size on the polarities}

Figure \ref{fig:yearly:pos2pos:changes} Shows the influence of embedding size on the cosine similarity between embeddings of Positive words with other classes (Positive, Negative, Random) words.\footnote{The effects are the same for Negative words.} 
The empirical results are mixed since on the one hand when the dimension decreases, the similarity inside the class of positive words increase (that is good), but the similarity between positive and negative words increases too. The similarity to random words stays close to zero.

\paragraph{Conjecture on the influence of embedding size on frequentist synonyms.}
This observation is compatible with the idea that the lower the dimension to represent the language, the more difficult to memorize the differences in contexts of words. As a consequence, words with the same polarity will become more similar, but antonyms will become more similar too.

\begin{figure}[!h]
    \centering
    \includegraphics[width=\linewidth]{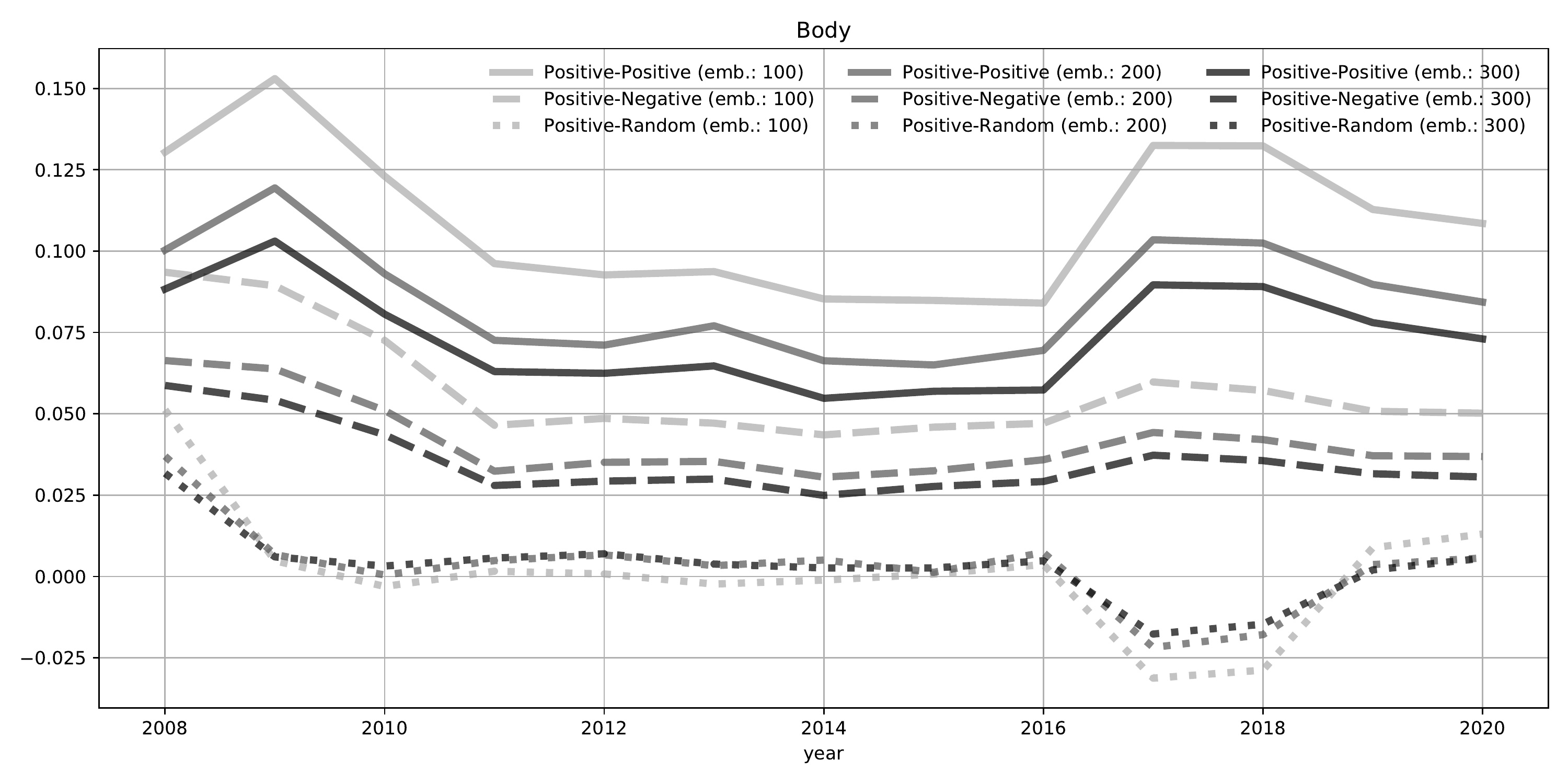}
    %{avg_cosin_group_yearly.pdf}
    \caption{Change of the cosine similarity between positive words and random words as a function of the size of embeddings.}
    \label{fig:yearly:pos2pos:changes}
\end{figure}

\subsubsection{When a company name becomes a frequentist synonym}

Figure \ref{fig:yearly:theranos} document an interesting effect that can have an influence on using embeddings to predict returns of listed companies.
The important question to ask is if we want that a sentence like
``\emph{Microsoft printed booming results}''\footnote{``booming'' is part of the Loughran-McDonald list of positively polarized words.} is understood the same way as ``\emph{Google printed booming results}'' by an embedding-based system.
It would require that \emph{Microsoft} and \emph{Google} have no polarity. 
There is no semantic reason for a company name to be positively or negatively biased, i.e. to have a cosine similarity different from zero to a group of polarized words.

This analysis is not systematic, we simply took the only company name of Wikipedia's \emph{List of corporate collapses and scandals} that is in our database, and computed its cosine similarity with the five most numerous lists of the Loughran-McDonald lexicon (Positive, Negative, Litigious and Uncertain).

% https://en.Wikipedia.org/wiki/List_of_corporate_collapses_and_scandals
\begin{figure}[!h]
    \centering
    \includegraphics[width=\linewidth]{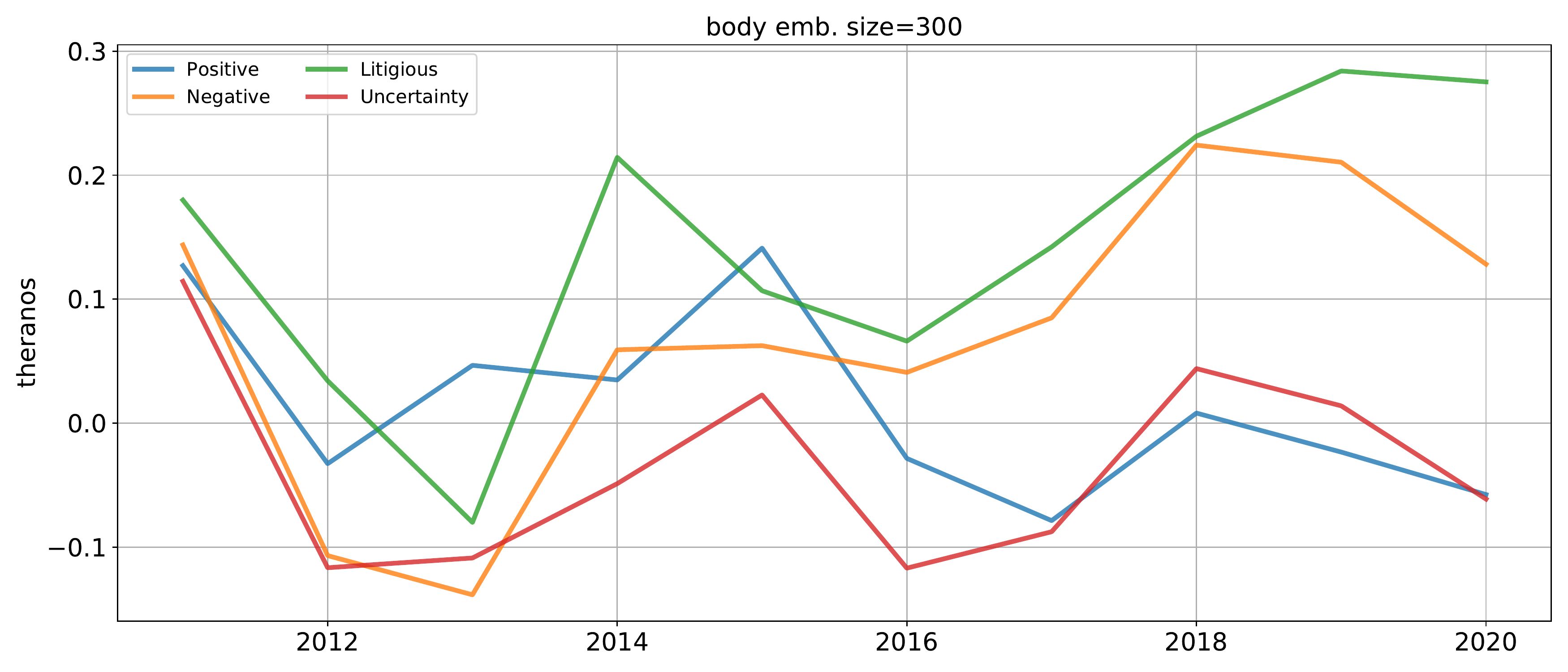}
    \caption{Yearly cosine similarity between the embeddings of Theranos and the Loughran-McDonald lexicon (it is measured using the {\tt model.wv.n\_similarity} function of gensim).}
    \label{fig:yearly:theranos}
\end{figure}

Before commenting the results, have a look at the quick summary made by Wikipedia\footnote{On Theranos Wikipedia page as of the 10th of March 2020.} of the Theranos case:
\\[.3em]
\framebox{\parbox{\linewidth}{\sf \small%
In March 2018 the US Securities and Exchange Commission charged Theranos, its CEO Elizabeth Holmes and former president Ramesh "Sunny" Balwani, claiming they had engaged in an "elaborate, years-long fraud" wherein they "deceived investors into believing that its key product – a portable blood analyzer – could conduct comprehensive blood tests from finger drops of blood".%
}}\\[.3em]

On Figure \ref{fig:yearly:theranos}, we see the change in the polarity of the embedding of the term \emph{Theranos} in the embeddings. This means that in 2018 and 2019, sentences with the name of this company is tinted with negativity and litigation compared to the same sentence concerning another company.
With an exaggerated anthropomorphism, we could say that the embedding now ``believes'' that most sentences concerning this company are negative or litigious. In fact this company name became a  frequentist synonym of negative and litigious terms.
Probably because from the viewpoint of skip-gram Word2vec loss function, it was ``easier'' to get this company name closer to negative and litigious terms than to keep it away from them (like any company name should semantically be).

The goal of this paper is not to investigate further on this kind of polarization of entities, that is in fact covered in a completely different context by the literature on \emph{fairness of NLP}, see \cite{romanov2019s} for an example and \cite{shah2019predictive} for an overview. 

% https://arxiv.org/pdf/1908.10063.pdf FinBERT

\section{Conclusion}

In this paper, we presented some theoretical understanding of word embeddings, essentially using the skip-gram word2vec model. Moreover we explain why more sophisticated model should inherit, at least locally, of some of these properties.
It allows us the define the concept of Reference Model (that is the uncompressed version of an embedding model), and to show that asymptotically the loss function of such a learning algorithm is a cross entropy between the representation of the model and the distribution of the Reference Model.
Moreover, it lead us to define frequentists synonyms, i.e. words that have the same context in the considered corpus. It is impossible for embeddings to make the difference between exact frequentist synonyms, and difficult for approximate synonyms.

Then we test these concepts on synthetic corpora generated using controlled Markovian models, so that we can focus on the identifiability of skip-gram word2vec embeddings. We observe that if their identifiability is poor, the cosine similarity between embeddings makes sense, even when it is low: frequentist synonyms are closer to word from their class than to words of another group of synonyms.

Last but not least, we provide empirical observations on a financial corpus: we use the Loughran-McDonald lexicon to obtain semantic synonyms and antonyms: lists of polarized words (Positive vs. negative words essentially).
And we compare the polarity of embeddings trained on headlines of News with other trained on the full text of the same News.
We observe that on News headlines, that are short and structured sentences, semantic antonyms are often frequentist synonyms, and hence it is difficult for embeddings learned on such headlines to make the difference between positive and negative words.
On the opposite, embeddings learned on the full body of the News are more reflecting the polarities of the considered lexicon. In fact they are better reflecting financial polarities than embeddings trained on Wikipedia.
It seems that the dimension of the embeddings has an influence on the cosine similarities between polarities: the lower the dimension, the more difficult to make the difference, in the space of embeddings, between positive and negative words.
We moreover observe that names of companies can be tinted with polarity. The structure of the loss function of embeddings can lead them to accept to represent a company name close to a polarized word if it appears a lot in negative News: this company name is thus becoming a frequentist synonym of negative and litigious words.

\bibliographystyle{unsrt}
\bibliography{maths_of_nlp}

\begin{thebibliography}{10}

\bibitem{laham1998learning}
TKLD Laham and Peter Foltz.
\newblock Learning human-like knowledge by singular value decomposition: A
  progress report.
\newblock {\em Advances in neural information processing systems}, 10:45, 1998.

\bibitem{lagus1999websom}
Krista Lagus, Timo Honkela, Samuel Kaski, and Teuvo Kohonen.
\newblock Websom for textual data mining.
\newblock {\em Artificial Intelligence Review}, 13(5):345--364, 1999.

\bibitem{2013arXiv1310.4546M}
Tomas Mikolov, Ilya Sutskever, Kai Chen, Greg~S Corrado, and Jeff Dean.
\newblock Distributed representations of words and phrases and their
  compositionality.
\newblock In C.~J.~C. Burges, L.~Bottou, M.~Welling, Z.~Ghahramani, and K.~Q.
  Weinberger, editors, {\em Advances in Neural Information Processing Systems},
  volume~26, pages 3111--3119. Curran Associates, Inc., 2013.

\bibitem{ghosh2016contextual}
Shalini Ghosh, Oriol Vinyals, Brian Strope, Scott Roy, Tom Dean, and Larry
  Heck.
\newblock Contextual lstm (clstm) models for large scale nlp tasks.
\newblock {\em arXiv preprint arXiv:1602.06291, presented at KDD 2016}, 2016.

\bibitem{vaswani2017attention}
Ashish Vaswani, Noam Shazeer, Niki Parmar, Jakob Uszkoreit, Llion Jones,
  Aidan~N Gomez, {\L}ukasz Kaiser, and Illia Polosukhin.
\newblock Attention is all you need.
\newblock In {\em Proceedings of the 31st International Conference on Neural
  Information Processing Systems}, pages 6000--6010, 2017.

\bibitem{devlin-etal-2019-bert}
Jacob Devlin, Ming-Wei Chang, Kenton Lee, and Kristina Toutanova.
\newblock {BERT}: Pre-training of deep bidirectional transformers for language
  understanding.
\newblock In {\em Proceedings of the 2019 Conference of the North {A}merican
  Chapter of the Association for Computational Linguistics: Human Language
  Technologies, Volume 1 (Long and Short Papers)}, pages 4171--4186,
  Minneapolis, Minnesota, June 2019. Association for Computational Linguistics.

\bibitem{implicitmatrix}
Omer Levy and Yoav Goldberg.
\newblock Neural word embedding as implicit matrix factorization.
\newblock In Z.~Ghahramani, M.~Welling, C.~Cortes, N.~Lawrence, and K.~Q.
  Weinberger, editors, {\em Advances in Neural Information Processing Systems},
  volume~27, pages 2177--2185. Curran Associates, Inc., 2014.

\bibitem{blei2003latent}
David~M Blei, Andrew~Y Ng, and Michael~I Jordan.
\newblock Latent dirichlet allocation.
\newblock {\em the Journal of machine Learning research}, 3:993--1022, 2003.

\bibitem{may2019downstream}
Avner May, Jian Zhang, Tri Dao, and Christopher R{\'e}.
\newblock On the downstream performance of compressed word embeddings.
\newblock {\em Advances in neural information processing systems}, 32:11782,
  2019.

\bibitem{raunak2019effective}
Vikas Raunak, Vivek Gupta, and Florian Metze.
\newblock Effective dimensionality reduction for word embeddings.
\newblock In {\em Proceedings of the 4th Workshop on Representation Learning
  for NLP (RepL4NLP-2019)}, pages 235--243, 2019.

\bibitem{acharya2019online}
Anish Acharya, Rahul Goel, Angeliki Metallinou, and Inderjit Dhillon.
\newblock Online embedding compression for text classification using low rank
  matrix factorization.
\newblock In {\em Proceedings of the AAAI Conference on Artificial
  Intelligence}, volume~33, pages 6196--6203, 2019.

\bibitem{loughran2011liability}
Tim Loughran and Bill McDonald.
\newblock When is a liability not a liability? textual analysis, dictionaries,
  and 10-ks.
\newblock {\em The Journal of finance}, 66(1):35--65, 2011.

\bibitem{kumar2017iitpb}
Abhishek Kumar, Abhishek Sethi, Md~Shad Akhtar, Asif Ekbal, Chris Biemann, and
  Pushpak Bhattacharyya.
\newblock Iitpb at semeval-2017 task 5: Sentiment prediction in financial text.
\newblock In {\em Proceedings of the 11th International Workshop on Semantic
  Evaluation (SemEval-2017)}, pages 894--898, 2017.

\bibitem{ke2019predicting}
Zheng~Tracy Ke, Bryan~T Kelly, and Dacheng Xiu.
\newblock Predicting returns with text data.
\newblock Technical report, National Bureau of Economic Research, 2019.

\bibitem{li2014news}
Xiaodong Li, Haoran Xie, Li~Chen, Jianping Wang, and Xiaotie Deng.
\newblock News impact on stock price return via sentiment analysis.
\newblock {\em Knowledge-Based Systems}, 69:14--23, 2014.

\bibitem{xing2018natural}
Frank~Z Xing, Erik Cambria, and Roy~E Welsch.
\newblock Natural language based financial forecasting: a survey.
\newblock {\em Artificial Intelligence Review}, 50(1):49--73, 2018.

\bibitem{araci2019finbert}
Dogu Araci.
\newblock Finbert: Financial sentiment analysis with pre-trained language
  models.
\newblock {\em arXiv preprint arXiv:1908.10063}, 2019.

\bibitem{2013arXiv1301.3781M}
Tomas {Mikolov}, Kai {Chen}, Greg {Corrado}, and Jeffrey {Dean}.
\newblock {Efficient Estimation of Word Representations in Vector Space}.
\newblock {\em arXiv e-prints}, page arXiv:1301.3781, January 2013.

\bibitem{word2vec}
Xin {Rong}.
\newblock {word2vec Parameter Learning Explained}.
\newblock {\em arXiv e-prints}, page arXiv:1411.2738, November 2014.

\bibitem{Goodfellow-et-al-2016}
Ian Goodfellow, Yoshua Bengio, and Aaron Courville.
\newblock {\em Deep Learning}.
\newblock MIT Press, 2016.
\newblock \url{http://www.deeplearningbook.org}.

\bibitem{LDA}
David~M. Blei, Andrew~Y. Ng, and Michael~I. Jordan.
\newblock Latent dirichlet allocation.
\newblock {\em The Journal of Machine Learning Research}, 3(null):993–1022,
  March 2003.

\bibitem{revuz2008markov}
Daniel Revuz.
\newblock {\em Markov chains}.
\newblock Elsevier, 2008.

\bibitem{LevinPeresWilmer2006}
David~A. Levin, Yuval Peres, and Elizabeth~L. Wilmer.
\newblock {\em Markov chains and mixing times}.
\newblock American Mathematical Society, 2006.

\bibitem{legall}
Jean-François~Le Gall.
\newblock Intégration, probabilités et processus aléatoires, Septembre 2006.

\bibitem{mixting}
Justin Salez.
\newblock Temps de mélange des chaînes de markov.

\bibitem{tao2012topics}
Terence Tao.
\newblock {\em Topics in random matrix theory}, volume 132.
\newblock American Mathematical Soc., 2012.

\bibitem{choromanska2015loss}
Anna Choromanska, Mikael Henaff, Michael Mathieu, G{\'e}rard~Ben Arous, and
  Yann LeCun.
\newblock The loss surfaces of multilayer networks.
\newblock In {\em Artificial intelligence and statistics}, pages 192--204.
  PMLR, 2015.

\bibitem{dim}
Zi~Yin and Yuanyuan Shen.
\newblock On the dimensionality of word embedding.
\newblock In S.~Bengio, H.~Wallach, H.~Larochelle, K.~Grauman, N.~Cesa-Bianchi,
  and R.~Garnett, editors, {\em Advances in Neural Information Processing
  Systems}, volume~31, pages 887--898. Curran Associates, Inc., 2018.

\bibitem{gentzkow2019text}
Matthew Gentzkow, Bryan Kelly, and Matt Taddy.
\newblock Text as data.
\newblock {\em Journal of Economic Literature}, 57(3):535--74, 2019.

\bibitem{romanov2019s}
Alexey Romanov, Maria De-Arteaga, Hanna Wallach, Jennifer Chayes, Christian
  Borgs, Alexandra Chouldechova, Sahin Geyik, Krishnaram Kenthapadi, Anna
  Rumshisky, and Adam~Tauman Kalai.
\newblock What's in a name? reducing bias in bios without access to protected
  attributes.
\newblock {\em arXiv preprint arXiv:1904.05233}, 2019.

\bibitem{shah2019predictive}
Deven Shah, H~Andrew Schwartz, and Dirk Hovy.
\newblock Predictive biases in natural language processing models: A conceptual
  framework and overview.
\newblock {\em arXiv preprint arXiv:1912.11078}, 2019.

\bibitem{fares-etal-2017-word}
Murhaf Fares, Andrey Kutuzov, Stephan Oepen, and Erik Velldal.
\newblock Word vectors, reuse, and replicability: Towards a community
  repository of large-text resources.
\newblock In {\em Proceedings of the 21st Nordic Conference on Computational
  Linguistics}, pages 271--276, Gothenburg, Sweden, May 2017. Association for
  Computational Linguistics.

\end{thebibliography}

\appendix

\section{Technical details}
\label{sec:tech:appendix}

We use Python 3.7.6 and Gensim 3.8.3 to train our models on AWS with financial news. On financial news headlines, we set \texttt{min\_count=1} to include every word in headlines while we let \texttt{min\_count=5} as the default value of Gensim on financial news which means we only consider the words count more than 5 times. Our \texttt{epoque =1}, i.e., each training sample is used only 1 time. As the model is initialized by training sentences, our \texttt{total\_examples=model.corpus\_count}.

Training skip-gram word2vec on one year of News takes around 2h and on one year of headlines it takes around 1h30. We use {\tt ml.p2.xlarge} as of March 2021, i.e. 1  NVIDIA K80 GPU, 1 vCPU, and 64Go of RAM.

The model trained on English Wikipedia Dump of February 2017 uses Gensim Continuous skip-gram with no lemmatization. It is provided by Language Technology Group at the University of Oslo in NLPL word embeddings repository \cite{fares-etal-2017-word}.

%Our syntactic data are generated and trained on Colab and we use the default version of Python and Pytorch on Colab. 

\end{document}